\documentclass[12pt]{article}
\usepackage{amsbsy,amsfonts,amsmath,amssymb,amsthm,bm,breakcites,graphicx}
\usepackage{lineno,hyperref,mathptmx,subcaption}

\newcommand{\aaa}{{\tt aaa}}
\newcommand{\abb}{{\tt abb}}
\newcommand{\abc}{{\tt abc}}
\newcommand{\acc}{{\tt acc}}
\newcommand{\bbd}{{\tt bbd}}
\newcommand{\bbe}{{\tt bbe}}
\newcommand{\bcd}{{\tt bcd}}
\newcommand{\bce}{{\tt bce}}
\newcommand{\ddd}{{\tt ddd}}
\newcommand{\ccd}{{\tt ccd}}
\newcommand{\dde}{{\tt dde}}
\newcommand{\Om}{\Omega}
\newcommand{\Si}{\Sigma}
\renewcommand{\H}{\mathbb{H}}
\newcommand{\R}{\mathbb{R}}
\renewcommand{\S}{\mathbb{S}}
\newcommand{\A}{{\cal A}}
\newcommand{\B}{{\cal B}}
\newcommand{\C}{{\cal C}}
\renewcommand{\O}{{\cal O}}
\newcommand{\T}{{\cal T}}
\newcommand{\cN}{{\cal N}}
\newcommand{\cS}{{\cal S}}
\newcommand{\cfk}{\mathfrak{c}}
\newcommand{\eps}{\varepsilon}
\newcommand{\Eps}{\bm{\eps}}
\newcommand{\ld}{\lambda}
\newcommand{\tet}{\theta}

\newcommand{\deh}{\partial}

\newcommand{\BE}{\begin{equation}}  
\newcommand{\EE}{\end{equation}}

\newtheorem{asp}{Assumption}
\newtheorem{defn}{Definition}
\newtheorem{expl}{Example}
\newtheorem{lem}{Lemma}
\newtheorem{prop}{Proposition}
\newtheorem{thm}{Theorem}

\begin{document}

\centerline{\large{\bf The Isoperimetric Problem in a Lattice of $\H^3$}}

\bigskip

\centerline{
Guillermo Lobos\footnote{Maths Dept, UFSCar, rod. Washington Lu\'is Km 235, 13565-905 S\~ao Carlos-SP, Brazil, {\bf Email: lobos@dm.ufscar.br}}, 
Alvaro Hancco\footnote{UFT, av. Paraguai esq.c/ r. Uxiramas, Cimba, 77824-838 Aragua\'ina-TO, Brazil, {\bf E-mail: alvaroyucra@uft.edu.br}}, 
Val\'erio Ramos Batista\footnote{{\it Corresponding Author} at UFABC, av. dos Estados 5001, 09210-580 St Andr\'e-SP, Brazil, {\bf Email: vramos1970@gmail.com}}}

\begin{abstract}
The isoperimetric problem is one of the oldest in geometry and it consists of finding a surface of minimum area that encloses a given volume $V$. It is particularly important in physics because of its strong relation with stability, and this also involves the study of phenomena in non-Euclidean spaces. Of course, such spaces cannot be customized for lab experiments but we can resort to computational simulations, and one of the mostly used softwares for this purpose is the Surface Evolver. In this paper we use it to study the isoperimetric problem in a lattice of the three dimensional hyperbolic space. More precisely: up to isometries, there exists a unique tesselation of $\H^3$ by non-ideal cubes $\C$. Now let $\Om$ be a connected isoperimetric region inside the non-ideal hyperbolic cube $\C$. Under weak assumptions on graph and symmetry we find all numerical solutions $\Si=\deh\Om$ of the isoperimetric problem in $\C$.
\end{abstract}


\section{Introduction}
\label{intro}

The isoperimetric problem is one of the oldest in geometry and it consists of finding a surface of minimum area that encloses a given volume $V$. For the readers not familiar with this problem we briefly comment on a physical experiment to identify possible shapes of a soap bubble inside a box. See Fig.~\ref{intro}.

\begin{figure}[ht!]
\center
\includegraphics[scale=0.25]{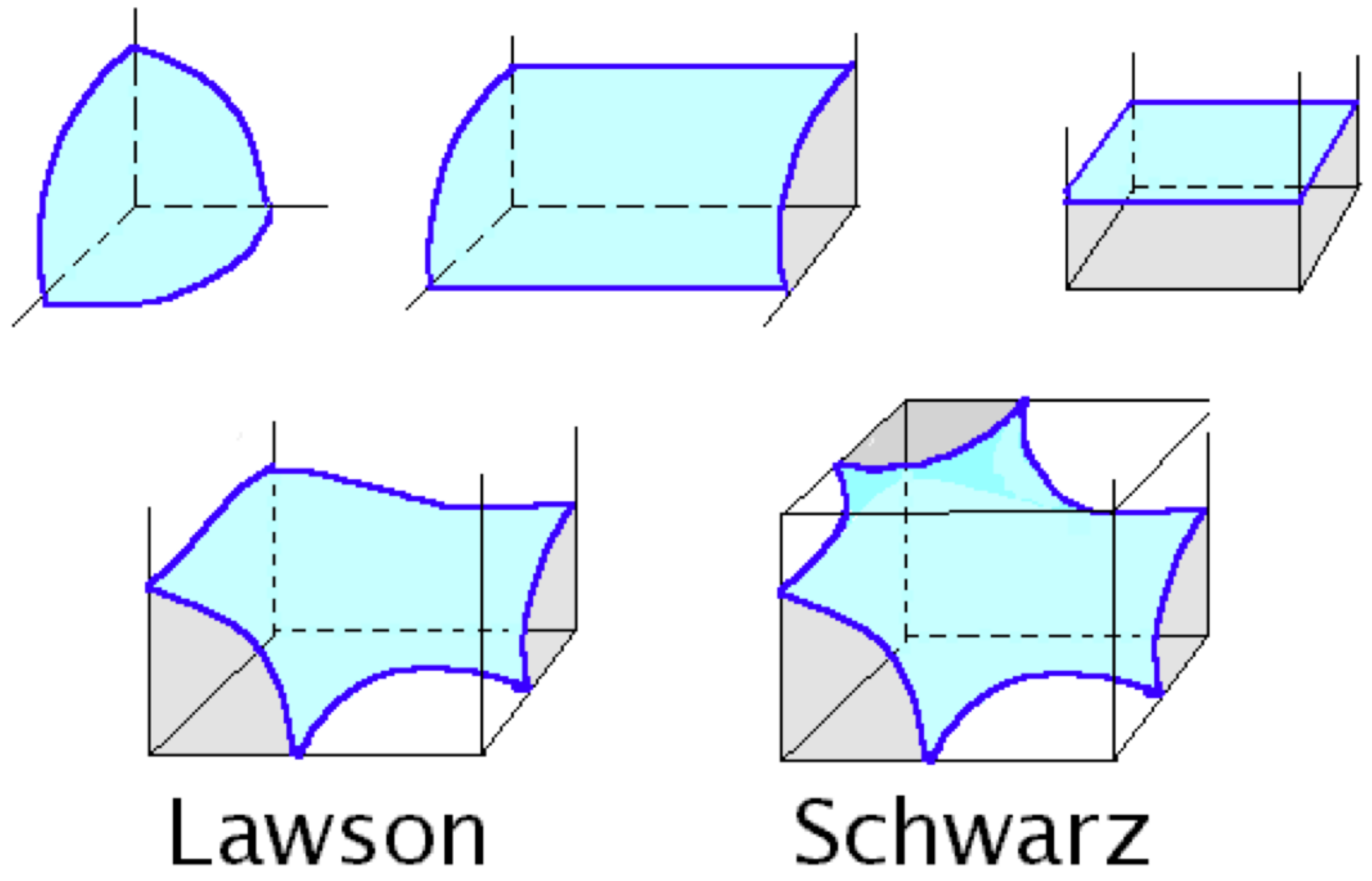}

\caption{Soap bubbles in a box \cite[Fig.9]{R1}.}

\label{intro}
\end{figure}

In Fig.~\ref{intro} one must consider that tilting the box will not change the shape of the bubble, and gravity is negligible because the soap film weighs close to zero. On top of Fig.~\ref{intro} we can easily see the three shapes that appear in a lab experiment: spherical, cylindrical and planar. Bottom right we have a surface named after the German mathematician Karl Hermann Amandus Schwarz (1843-1921) but this shape will never occur in the experiment as proved in \cite{Ri}. 

Bottom left we see another surface, this one found by Herbert Blaine Lawson (1942-). Though it has never appeared in the bubble experiment there is no mathematical proof that discards Lawson surface. Maybe this bubble could exist under special conditions but it has been an open question for two decades already. We shall resume it with details in Sect.~\ref{res}.

The isoperimetric problem is particularly important in physics because of its strong relation with stability, and this also involves the study of phenomena in non-Euclidean spaces. For example, by considering homogeneous density we can work with either volume $V$ or mass $m$. In \cite{P} the mathematical physicist Roger Penrose (1931-) combined several results and evidences regarding gravitational collapse in order to conjecture that, if $m$ is the total mass and $A$ the area of a black hole, then $4m\sqrt{\pi}\ge A$. In the theory of General Relativity this and other inequalities are called {\it isoperimetric inequalities for black holes}~\cite{Gi}. For this problem the non-Euclidean spaces of interest are the Schwarzschild and the Reissner-Nordstrom, and here we cite \cite{C} for the readers who want more details.

Differently from the example in Fig.~\ref{intro}, of course one cannot customize lab experiments with a non-Euclidean metric. But we can resort to computational simulations, and one of the mostly used softwares for this purpose is the {\it Surface Evolver}. Firstly introduced in 1989, now its most recent version is 2.70 \cite{Ken} with several applications in many Areas of Knowledge like Aerodynamics \cite{Frank}, Fluid Dynamics \cite{Chen}, \cite{Cunsolo}, \cite{Caio}, and Medicine \cite{Zanka}. It handles forces (contact, gravity, etc.), pressures, densities, $n$-dimensional spaces (including non-Euclidean), tracking of quantities and prescribed energies, among several other features.

In this paper we use {\it Evolver} to study the experiment of Fig.~\ref{intro} in a compact hyperbolic cubic box obtained by a tesselation of the whole space. Here we represent the hyperbolic 3-space $\H^3$ by the unit ball in $\R^3$ centred at the origin and endowed with the Poincar\'e metric. It is known that, up to isometries, there exists a unique tesselation of $\H^3$ by non-ideal cubes $\C$. This result is proved in Sect.~\ref{prelim}, which also includes two weak assumptions on isoperimetric regions inside $\C$. The assumptions come from the fact that any torus $\H^3/\Gamma$, given by a group of translations $\Gamma$ in $\H^3$, will be too poor in space symmetries compared with a Euclidean torus $\R^3/G$, where $G$ is a group of translations in $\R^3$. 

\section{Preliminaries}
\label{prelim}

This section is devoted to main theorems and propositions used throughout the text.

\begin{prop}
Up to isometries there exists a unique tesselation of $\H^3$ by non-ideal cubes $\C$. These cubes have dihedral angle $2\pi/5$, meaning that $n=5$ cubes meet at each edge.
\label{eu}
\end{prop}
\begin{proof}
For the uniqueness pick a vertex $p$ of the tessellation incident to $w$ edges. Now consider the unit sphere $S^2$ of the tangent space $T_p\H^3$. This space is isometric to $\H^3$, whose origin is now $p$. Of the $w$ edges each pair that belongs to a face of the cubic tesselation will make this face intersect with $S^2$. These intersections determine a triangulation of $S^2$ that is both equilateral and equiangular, hence one of only three possible types: tetrahedral ($n=3$), octahedral ($n=4$), and icosahedral ($n=5$). However, any hyperbolic cube has an acute dihedral angle. Therefore, only the case $n=5$ is possible.

Now define a ¯cube of radius $R$ by placing vertices on the 12 geodesic rays with icosahedral symmetry emanating from $p$. The dihedral angle can be computed in terms of hyperbolic trigonometry: it is a continuous and strictly decreasing function of~$R$. This function converges to $2\pi/4=\pi/2$ for $R\to 0$, and to $2\pi/6=\pi/3$ for $R\to\infty$. The {\it Intermediate Value Theorem} gives a cube $\C$ with dihedral angle $2\pi/5$.
\end{proof}

\begin{lem}
Let $D$ be the unit disk in $\R^2$ centred at the origin. Inversion with respect to a circumference of radius $r$ centred in $(c,0)$ keeps $D$ if and only if $c^2=1+r^2$.
\label{inv}
\end{lem}
\begin{proof}

Fig.~\ref{inv} illustrates $D$ in light grey and the circumference $c+r\,e^{i\theta}$, $0\le\theta<2\pi$. By taking $\R^2$ as the complex plane we know that such inversion is
\BE
   z\mapsto\frac{r^2}{\bar{z}-c}+c.
   \label{moe}
\EE

\begin{figure}[ht!]
\centering
 \includegraphics[scale=0.5]{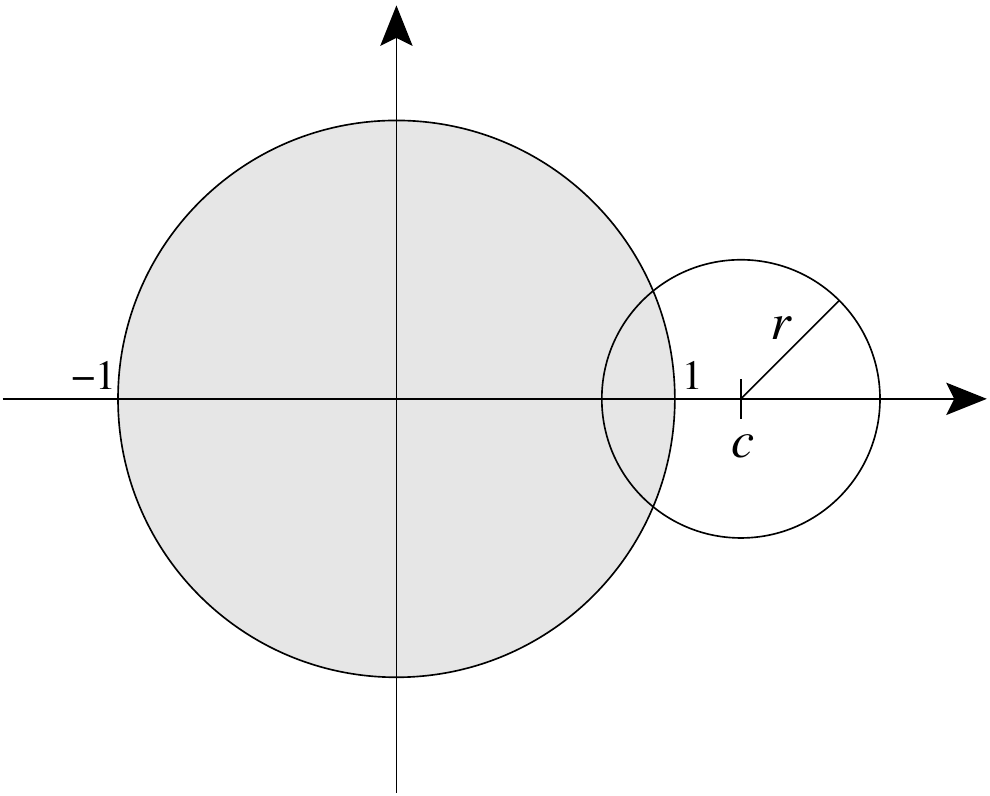}
 \caption{The circle $D$ and the circumference of inversion.}
\label{inv}
\end{figure} 

In order to keep $D$ it must interchange $1$ and $-1$, hence $c^2=1+r^2$. Notice that~(\ref{moe}) takes the origin to $1/c\in D$. Conversely, if $c^2=1+r^2$ then~(\ref{moe}) becomes $(c\bar{z}-1)/(\bar{z}-c)$, which obviously keeps $\deh D$. But since we wave $0\mapsto1/c<1$ then it keeps $D$.
\end{proof}
  
Let us define a lattice of isometries in $\H^3$ which keeps the cubic tessellation. Consider the three pairs of opposite faces of $\C$. For each such pair there is a mirror reflection $\sigma_1$, $\sigma_2$, $\sigma_3$ in one of the faces, and there are also mirror reflections $\tau_1$, $\tau_2$, $\tau_3$ which flip only the pair of opposite faces (and fix the midpoint of $\C$). The latters are represented by mirror reflections in the three coordinate planes in Fig.~\ref{3planes&tiling45}(a). From Lemma~\ref{inv} the formers are represented by spherical inversions such as 
\[
   q\mapsto r^2\cdot\frac{q-(c,0,0)}{||q-(c,0,0)||^2}+(c,0,0),
\]
where $c=(\sqrt{5}+2)^{1/2}$ and $r=(\sqrt{5}+1)^{1/2}$. These are the values that keep the tiling $\{4,5\}$ depicted in Fig.~\ref{3planes&tiling45}(b). Now we compose $T_i:=\tau_i\circ\sigma_i$ to obtain a ``translation'', namely an isometry of hyperbolic type. Finally we let $\Gamma:=\langle T_1,T_2,T_3\rangle$ and $\T:=\H^3/\Gamma$. Then $\C$ is a fundamental domain for $\T$.

\begin{figure}[ht!]
\centering
 \begin{subfigure}[ht!]{0.5\textwidth}
 \centering
 \includegraphics[scale=0.49]{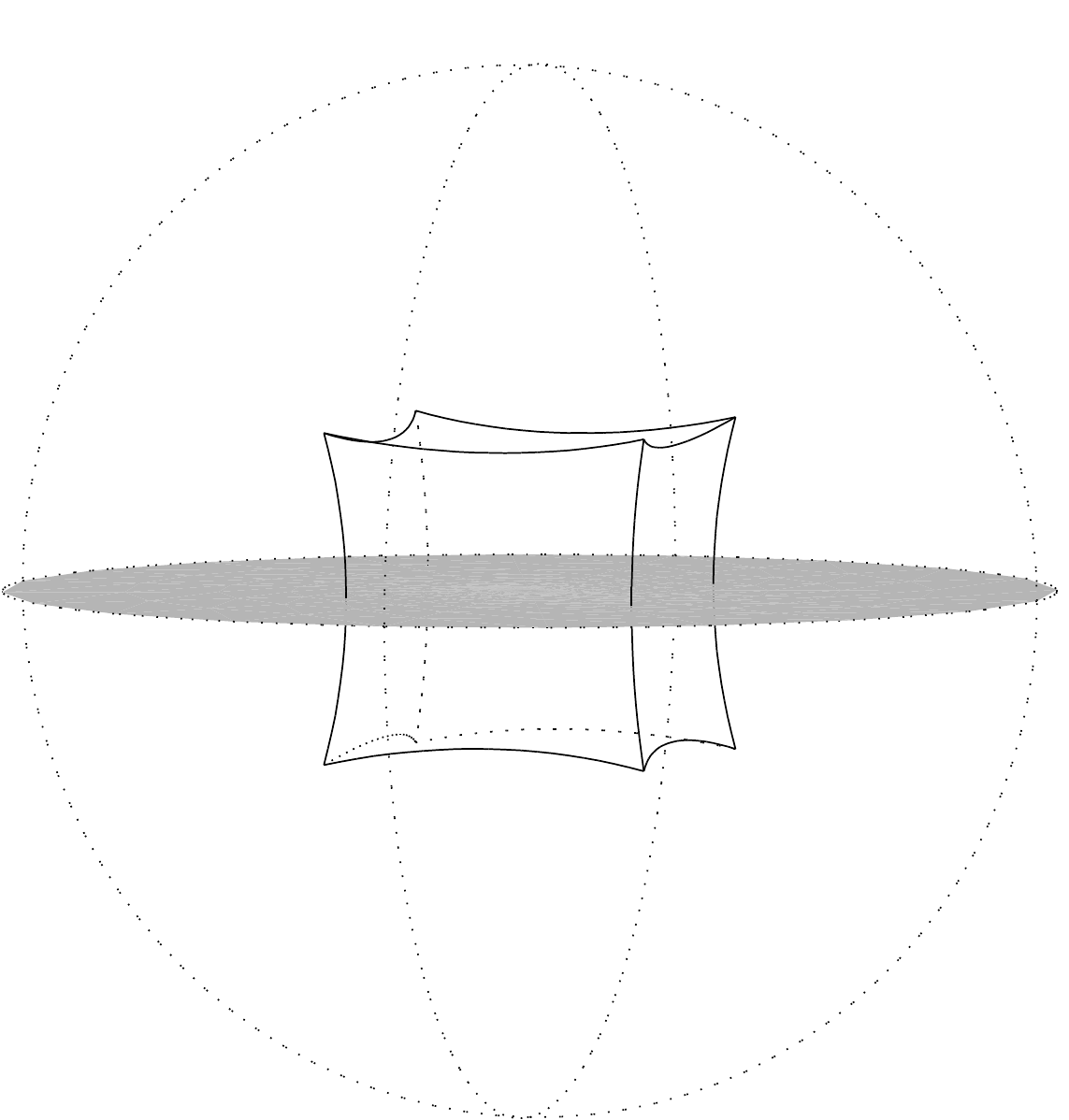}
 \caption{}
 \end{subfigure}%
 ~ 
 \begin{subfigure}[ht!]{0.5\textwidth}
 \centering
 \includegraphics[scale=0.34]{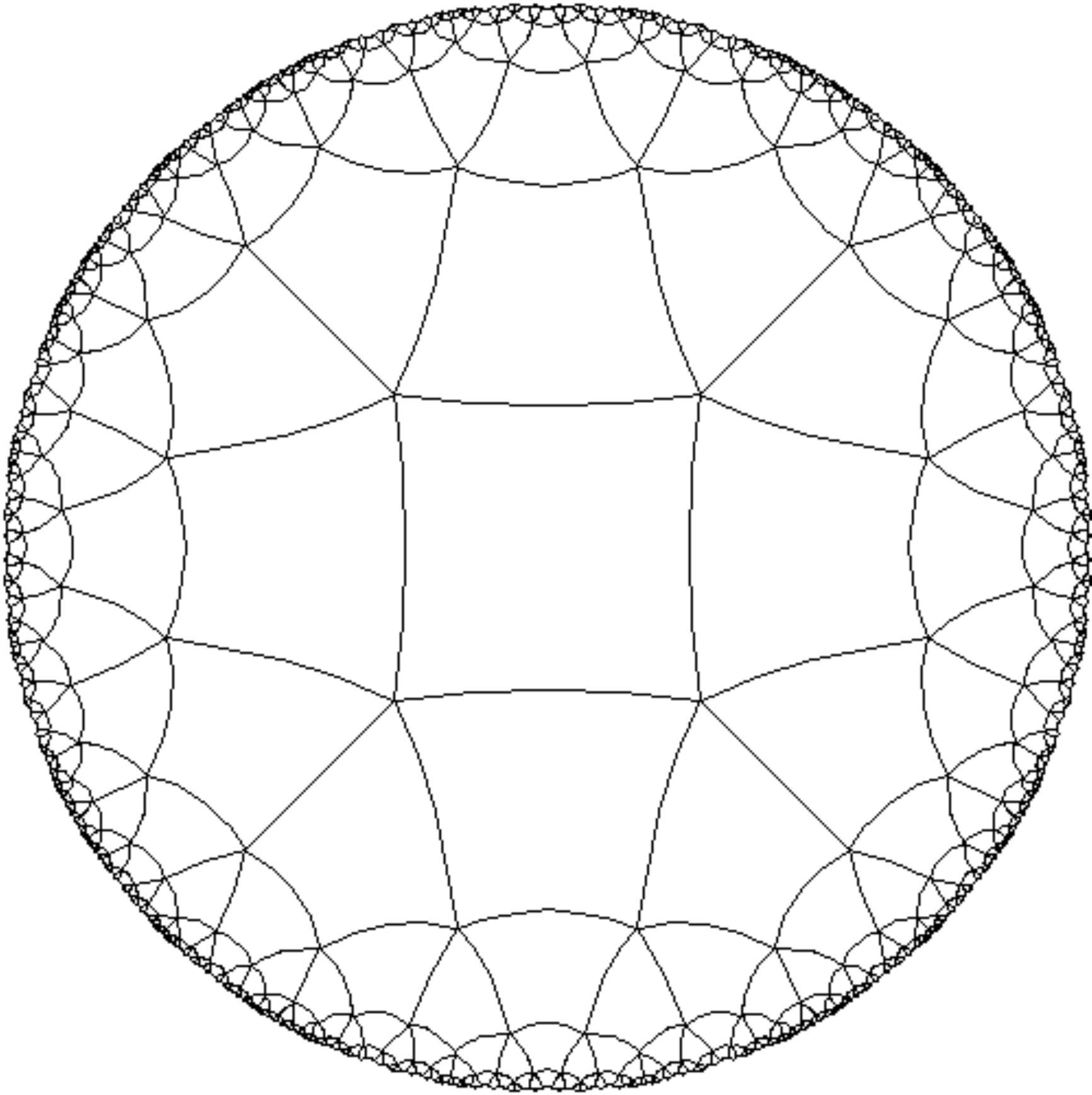}
 \caption{}
 \end{subfigure}
\caption{The cube $\C$ (a); the tiling $\{4,5\}$ (b).\protect\footnotemark}
\label{3planes&tiling45}
\end{figure} 

By means of the {\it Alexandrov Reflection Principle}~\cite{Al} we can reduce our study of the isoperimetric problem to a symmetric eighth $\B\subset\C$. In the Poincar\'e model if $\C$ is as in Fig.~\ref{3planes&tiling45}(a) then $\B:=\C\cap\{x_1\ge 0\}\cap\{x_2\ge 0\}\cap\{x_3\ge 0\}$. If $V$ denotes hyperbolic volume then $V(\C)\cong1.723$ and $V(\B)\cong0.215$. These values were obtained through the {\it Surface Evolver}, whose details are given in Sect.~\ref{ev}.

\footnotetext{From \url{https://mathcs.clarku.edu/~djoyce/poincare/tilings.html}}

The conformal metric in $\H^3$ is given by
\BE
   \frac{4\,I_3}{(1-x_1^2-x_2^2-x_3^2)^2},
   \label{confm}
\EE
where $I_3$ is the $3\times3$ identity matrix. For instance, if we consider a Euclidean radius $\eps\in(0,1)$ the area of and volume inside $S_\eps\subset\H^3$ are given by 
\BE
   A(S_\eps)=\frac{16\pi\eps^2}{(1-\eps^2)^2}
   \hspace{1cm}{\rm and}\hspace{1cm}
   V(S_\eps)=2\pi\biggl[\frac{2\eps(1+\eps^2)}{(1-\eps^2)^2}+\ln\frac{1-\eps}{1+\eps}\biggl],
   \label{av}
\EE
respectively. Notice that both $A(S_\eps)$ and $V(S_\eps)$ are strictly increasing with $\eps$. As mentioned at the Introduction, in this work we adopt two assumptions on isoperimetric regions in $\B$.

\begin{asp}
Any isoperimetric region in $\B$ intersects the three coordinate planes.
\label{assp1}
\end{asp}
 
As a matter of fact we believe that Assumption~\ref{assp1} is always true but were not able to prove this fact yet. Up to ambient isometries it is valid in Euclidean three-dimensional boxes, as proved in~\cite{Ri}. For now here is an example that motivates Assumption~\ref{assp1}:

\begin{expl}
The set $T_1:=S_\eps\cap\B$ encloses volume $V(S_\eps)/8$. Now let $S$ be a sphere centred at the upper right corner of the central square in Fig.~\ref{3planes&tiling45}(b), which in Euclidean coordinates corresponds to $(\cfk,\cfk,0)$ where
\BE
   \cfk=\frac{\sqrt{\smash[b]{\sqrt{5}+2}}-\sqrt[4]{5}}{2}.
   \label{corn}
\EE
If $T_2:=S\cap\B$ encloses the {\it same} volume $V(S_\eps)/8$ then $V(S)=5V(S_\eps)/4$ and therefore $A(T_2)>A(T_1)$.
\label{exsph}
\end{expl}

Now we present a strong result that will be widely used in our work. According to \cite{R1}, from \cite{A}, \cite{GMT}, \cite{G}, \cite{M} we have:

\begin{thm} Suppose $M^3\subset\C$ is compact and $\deh M$ is either empty or piecewise smooth. Then for any $t\in(0,V(M))$ there exists a compact domain $\Om\subset M$ such that $\Si=\deh\Om\setminus\deh M$ minimizes area among regions of volume~$t$. Moreover, the boundary of any minimizing region is a smooth embedded surface with constant mean curvature and, if $\deh M\cap\Si\ne\emptyset$, then $\Si$ meets $\deh M$ orthogonally.
\label{mainthm}
\end{thm}

Any $\Si$ that minimizes area under volume constraint has {\it constant mean curvature} (CMC). For instance, see \cite{BC} to check this well-known property. In their turn CMC surfaces follow the {\it maximum principle} (see \cite{GT}), which is the key to conclude that $\deh M\perp\Si$ whenever $\deh M\cap\Si\ne\emptyset$. The property of orthogonal intersections will be used extensively in this work. The following result is an important clue on Assumption~\ref{assp1}:

\begin{prop}
Let $\Om$ be an isoperimetric region of $\B$ and $\Si=\deh\Om\setminus\deh\B$. Then $\deh\Om\setminus\Si$ cannot contain open subsets of only zero, one, or two incident faces of $\deh\B$. In particular we have $\deh\B\cap\Si\ne\emptyset$.  
\label{clue}
\end{prop}
\begin{proof}
Let us apply Theorem~\ref{mainthm} to $M=\B$. By contradiction, consider a continuous family $\{\phi_t:t\in[0,1]\}$ of hyperbolic isometries that leave invariant the zero, one, or two coordinate planes. Let $\phi_0$ be the identity and $\phi_1(\B)\cap\B=\emptyset$. Now if $\Om$ intersects only zero, one, or two incident faces, then we consider for the respective family $\phi_t$ the supremum $T$ of the values $t$ such that $\phi_t(\Om)\subset\B$. Since $\phi_T$ keeps $V(\Om)$ then $\phi_T(\Om)$ is also minimizing. However, $\phi_T(\Om)$ touches $\deh\B$ tangentially, and so violates $\deh M\perp\Si$.
\end{proof}

Because of Proposition~\ref{clue} successive reflections of $\Si$ in the coordinate planes will result in a potential solution of the isoperimetric problem for $t=8V(\Om)$ in the torus $\T$ (equivalently in its fundamental domain $\C$).

\begin{asp}
Let $\Si$ be as in Proposition~\ref{clue}. Then its intersection with any coordinate plane is either empty or a two-dimensional connected graph.
\label{assp2}
\end{asp}

In this paper we analyse $\Si$ under Assumptions~\ref{assp1} and~\ref{assp2}. A little reflection shows that for $1\le i<j\le3$ there are only four non-empty possible connected topological cases depicted in Fig.~\ref{cases}. Notice that Fig.~\ref{cases}(b) is a graph in both directions since our geometry is hyperbolic. Moreover, we shall see that the numerical candidates for this case are either unduloids or a pair of equidistant tori. 

\begin{figure}[ht!]
\center
\includegraphics[scale=0.65]{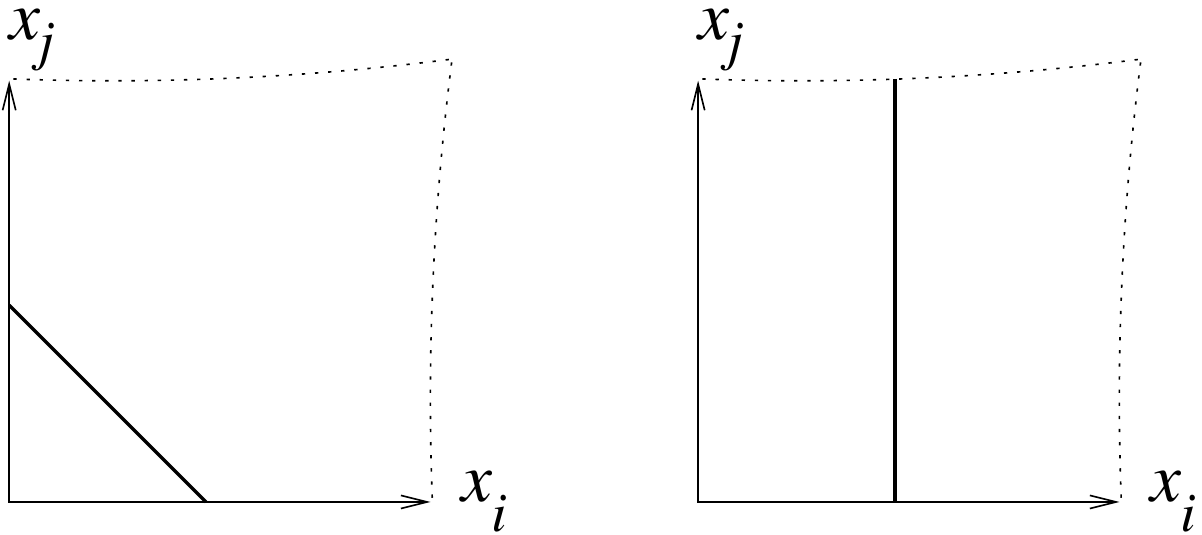}

(a)\hspace{4cm}(b)\\
\includegraphics[scale=0.65]{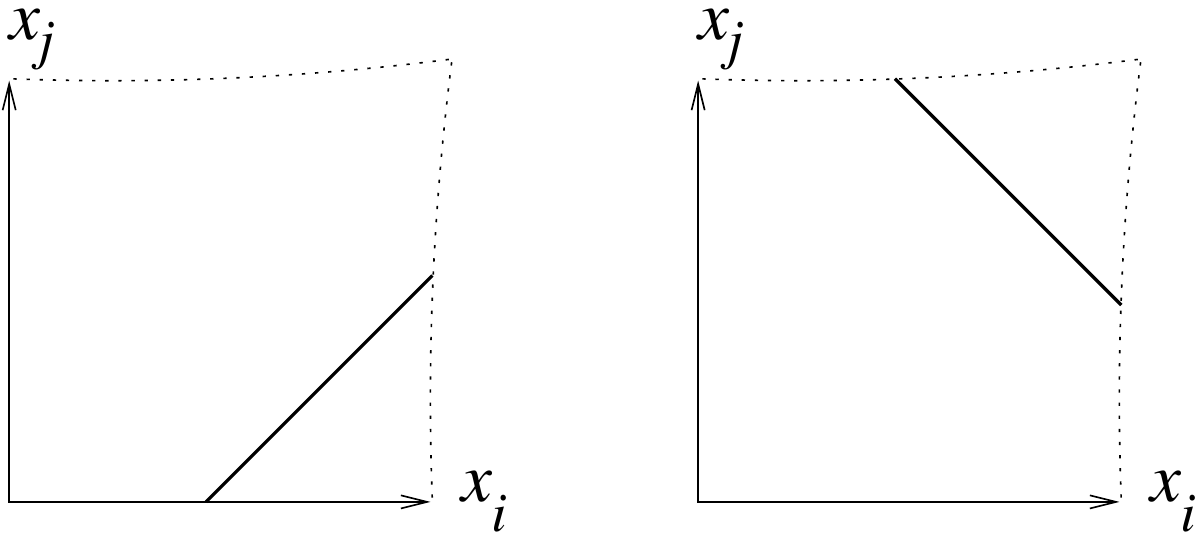}

(c)\hspace{4cm}(d)
\caption{The four possible non-empty graph types in $Ox_ix_j$.}
\label{cases}
\end{figure}

Not all combinations of the three axes are feasible. Table~\ref{candid} shows the eight possible combinations and the name we attribute for each one. The empty graph is called ``e'', the triad of letters $\ell_1\ell_2\ell_3$ corresponds to $Ox_1x_2$, $Ox_1x_3$ and $Ox_2x_3$ in this order, respectively. More details on Table~\ref{candid} will be given in Sect.~\ref{res}. There we explain why $\acc$ is called {\it inverted Lawson}, and why some cases like $\bce$, $\ccd$ and $\dde$ are not included. Notice that $\ell_1\ell_2\ell_3$ is always in alphabetical order because of congruence. For example, {\tt acb} is congruent to $\abc$.

The isoperimetric problem considers $V(\Om)\le V(\B)/2\cong0.108$. We could also take the origin $O\in\Om$ but for the sake of visibility some cases in Table~\ref{candid} are studied for $O\not\in\Om$.

\begin{table}
\centering
\caption{Triad of letters and some known examples.}
\begin{tabular}{|c|c|}\hline
Combination & Surface  \\ \hline\hline
$\aaa$      & sphere   \\
$\abb$      & unduloid \\
$\abc$      & \\
$\acc$      & inverted Lawson \\
$\bbd$      & Lawson \\
$\bbe$      & pair of tori \\
$\bcd$      & \\
$\ddd$      & Schwarz \\ \hline
\end{tabular}
\label{candid}
\end{table}

Before going ahead one should notice that Assumption~\ref{assp2} requires the graph to be connected. Otherwise one should consider another two extra cases depicted in Fig.~\ref{extra}. But later we shall see that they perform worse than the candidates listed in Table~\ref{candid}, which thus omitted them. 

\begin{figure}[ht!]
\center
\includegraphics[scale=0.65]{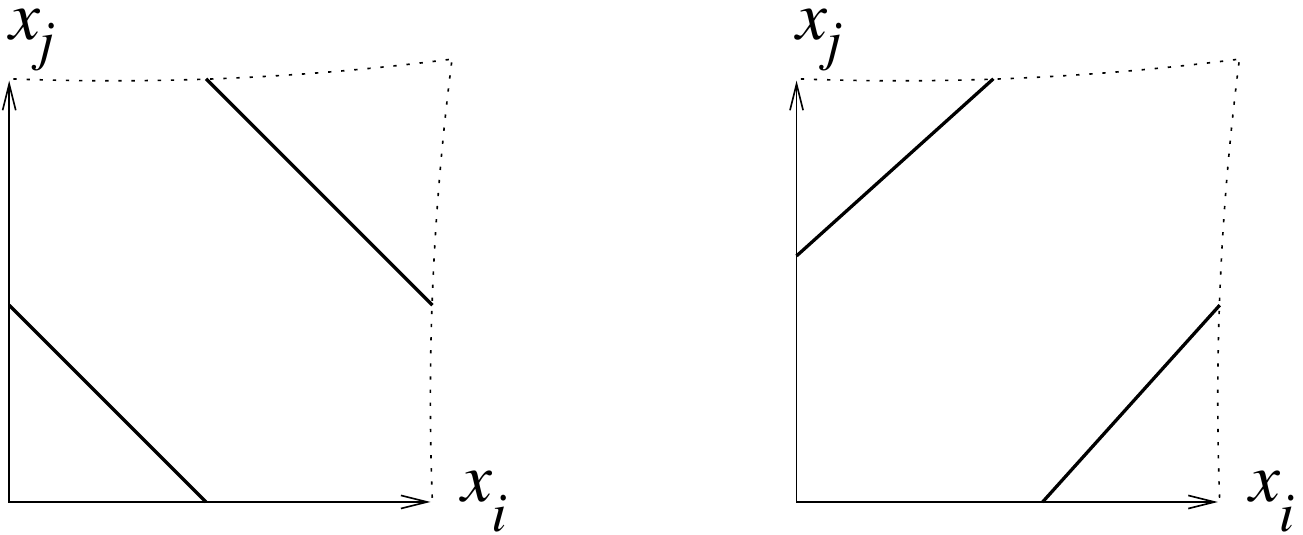}

\caption{The two possible non-connected graph types in $Ox_ix_j$.}
\label{extra}
\end{figure}

\section{Evolver Datafiles}
\label{ev}

In order to study each case numerically we make use of the {\it Surface Evolver} \cite{B1}, \cite{B2}. For lengths and areas {\it Evolver} computes Riemannian values through the metric given in the datafile. However, volumes must be achieved as a declared quantity and the computation does not use the given metric any longer. Therefore, we implemented the submanifold quantity {\tt sbmvol} in order to have
\BE
   \int_M1dV=\int_M\frac{8\,dv}{(1-x_1^2-x_2^2-x_3^2)^3}, 
   \label{sbmvol}
\EE
where $dV$ and $dv$ are the hyperbolic and Euclidean elements of volume, respectively.

But {\it Evolver} only computes {\it surface integrals} (line integrals are in fact sums of thin strips of surface). So we need to introduce a vector field {\tt (q1,q2,q3)} in such a way that its divergence is $8/(1-x_1^2-x_2^2-x_3^2)^3$. By setting $\rho=x_3/\sqrt{1-x_1^2-x_2^2}$ we may first compute
\BE
   \int\frac{8d\rho}{(1-\rho^2)^3}=
   \frac{2\rho}{(1-\rho^2)^2}+\frac{3\rho}{1-\rho^2}+\frac{3}{2}\ln\frac{1+\rho}{1-\rho}.
\EE

Hence, if we take {\tt q1}, {\tt q2} as identically zero then {\tt q3} is given by 
\[
   \int\frac{8dx_3}{(1-x_1^2-x_2^2-x_3^2)^3}=\frac{2x_3}{(1-x_1^2-x_2^2-x_3^2)^2(1-x_1^2-x_2^2)}+
\]
\BE
   \frac{3x_3}{(1-x_1^2-x_2^2-x_3^2)(1-x_1^2-x_2^2)^2}+
   \frac{3/2}{(1-x_1^2-x_2^2)^{5/2}}\ln\biggl|\frac{\sqrt{1-x_1^2-x_2^2}+x_3}{\sqrt{1-x_1^2-x_2^2}-x_3}\biggl|.
\label{q3}
\EE

In this way {\tt (q1,q2,q3)} are not symmetric. A better choice is to compute {\tt q1} and {\tt q2} as done for {\tt q3}, and then take one third of each expression. That is what we use in the datafiles.

Moreover, notice that Assumption~\ref{assp1} and Theorem~\ref{mainthm} enable us to study the isoperimetric regions $\Om\subset\B$ in such a way that the computation of the volume $V(\Om)$ will not use $\deh\Om\cap\deh\B$. For instance, the face of $\Si$ on ${x_3=0}$ makes (\ref{q3}) vanish, and the analogous holds for the symmetric {\tt (q1,q2,q3)} used in our datafiles. Hence we can work with $\Si=\deh\Om\setminus\deh\B$, as described in Proposition~\ref{clue}.

Inside $\C$ the sphere $S_\eps$ centred at the origin cannot have $\eps\ge\Eps=c-r\cong0.26$. Namely, for the first line of Table~\ref{candid} our simulation takes $\eps\in(0,\Eps)$. With {\it Evolver} 2.70 and {\it Geomview} 1.9.4 we obtained Fig.~\ref{aaa}. It depicts the numerical surface, whereas Fig.~\ref{aaa}(b) shows $V\times A$ for~(\ref{av}) and Fig.~\ref{aaa}(c) the corresponding $V\times\Delta A$, where $\Delta A$ is the difference between numerical and theoretical values of area.

In our simulation {\it Evolver} reached numerical values of volume that coincide with the theoretical ones up to the 8th decimal. Hence we consider them as identical. In Fig.~\ref{aaa}(c) we have $\Delta A\le 0.00163$, and the abrupt change at $V\cong0.02$ just means that we refined the triangulation for that volume. This is to guarantee that $A/F\le0.01$, where $F$ is the number of triangles. The fact that $V\le V_{\Eps}\cong0.083$ can be seen in Fig.~\ref{aaa}(b), namely lesser than $V(\B)/2\cong0.108$, which means a great difference from the Euclidean case. Indeed, if that eighth of cube were Euclidean the extreme $\aaa$ would have more than half of its volume.

\begin{figure}[ht!]
\center
\includegraphics[scale=0.27]{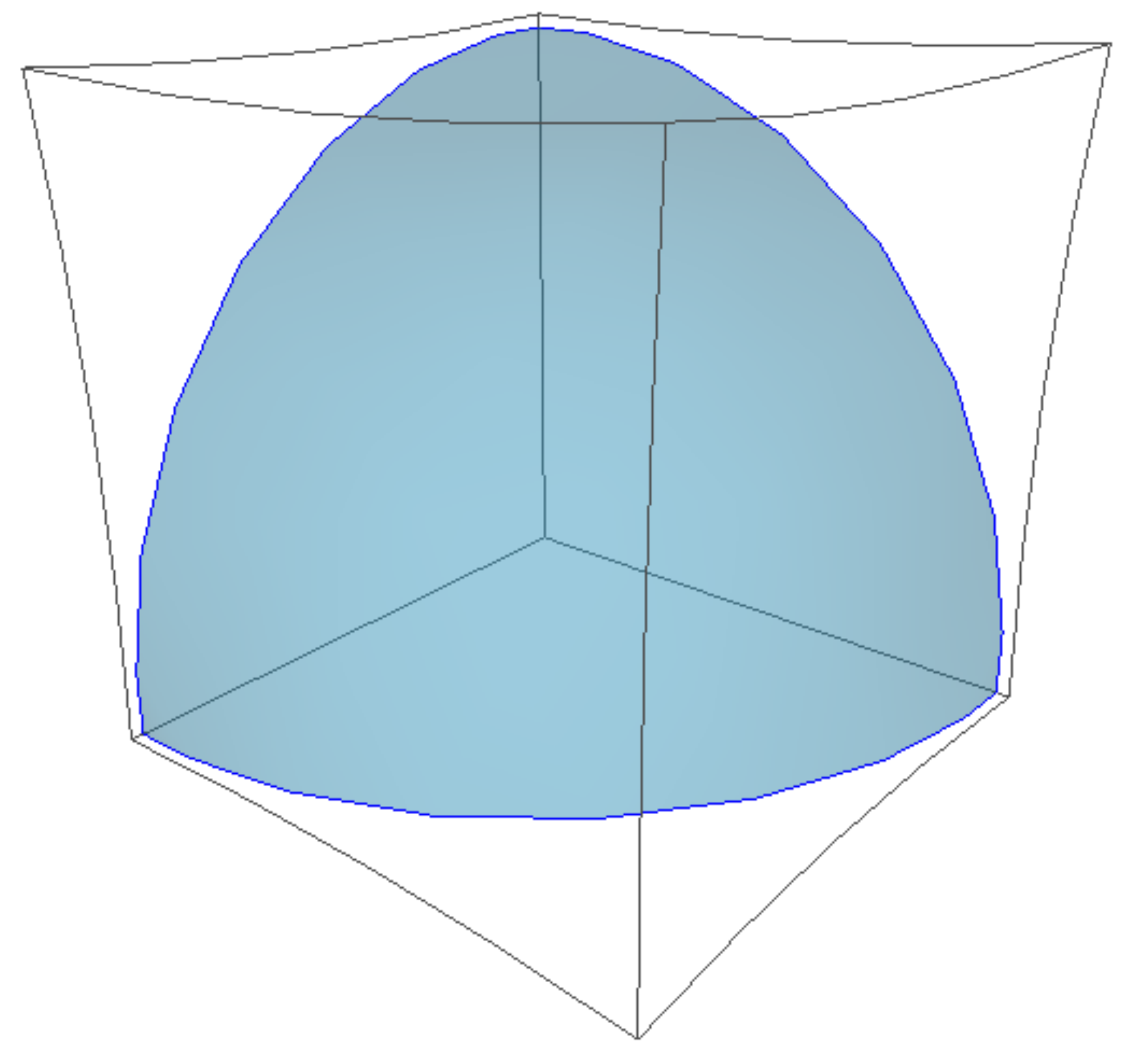}

(a)\\

\smallskip

\includegraphics[scale=0.32]{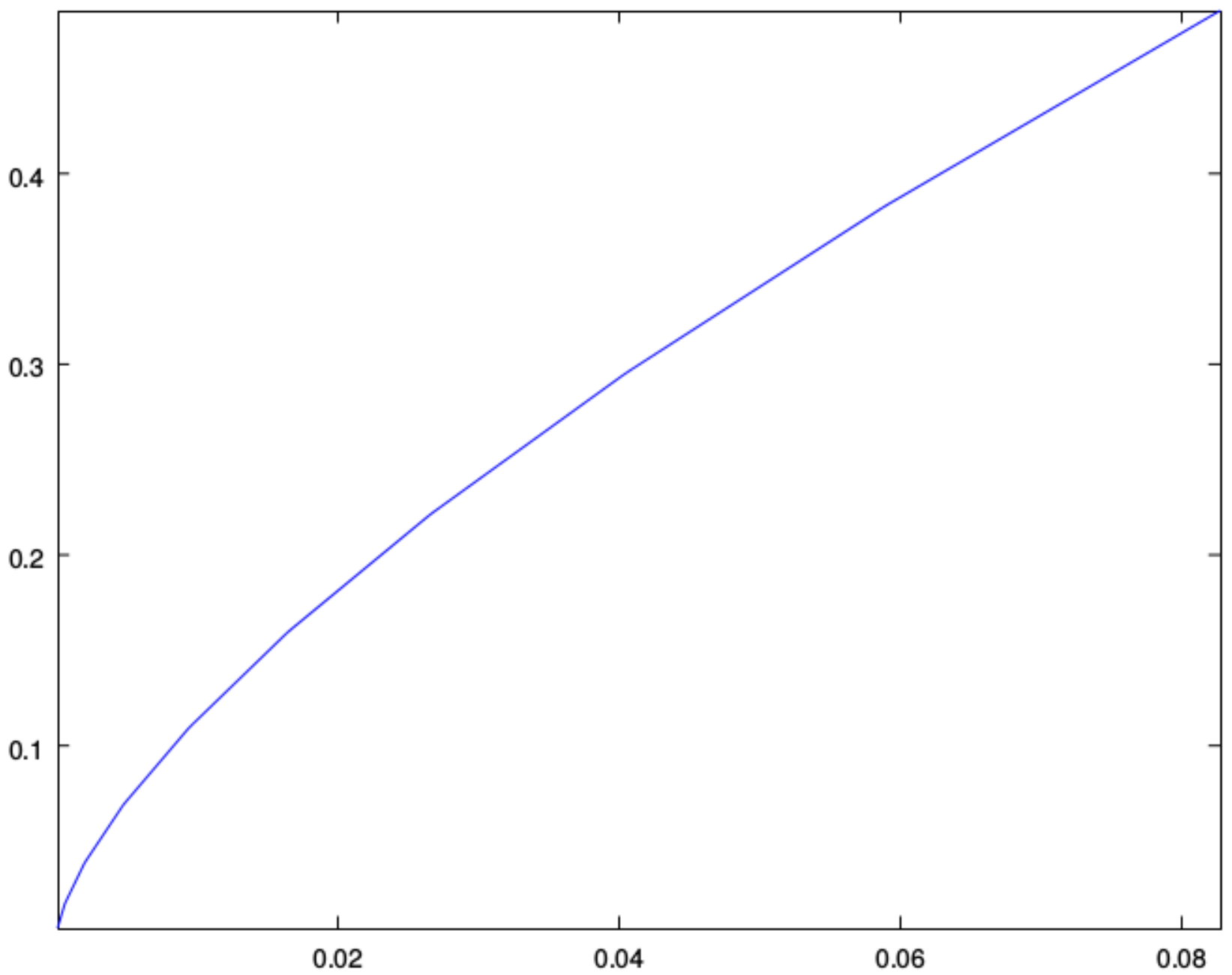}\hfill
\includegraphics[scale=0.32]{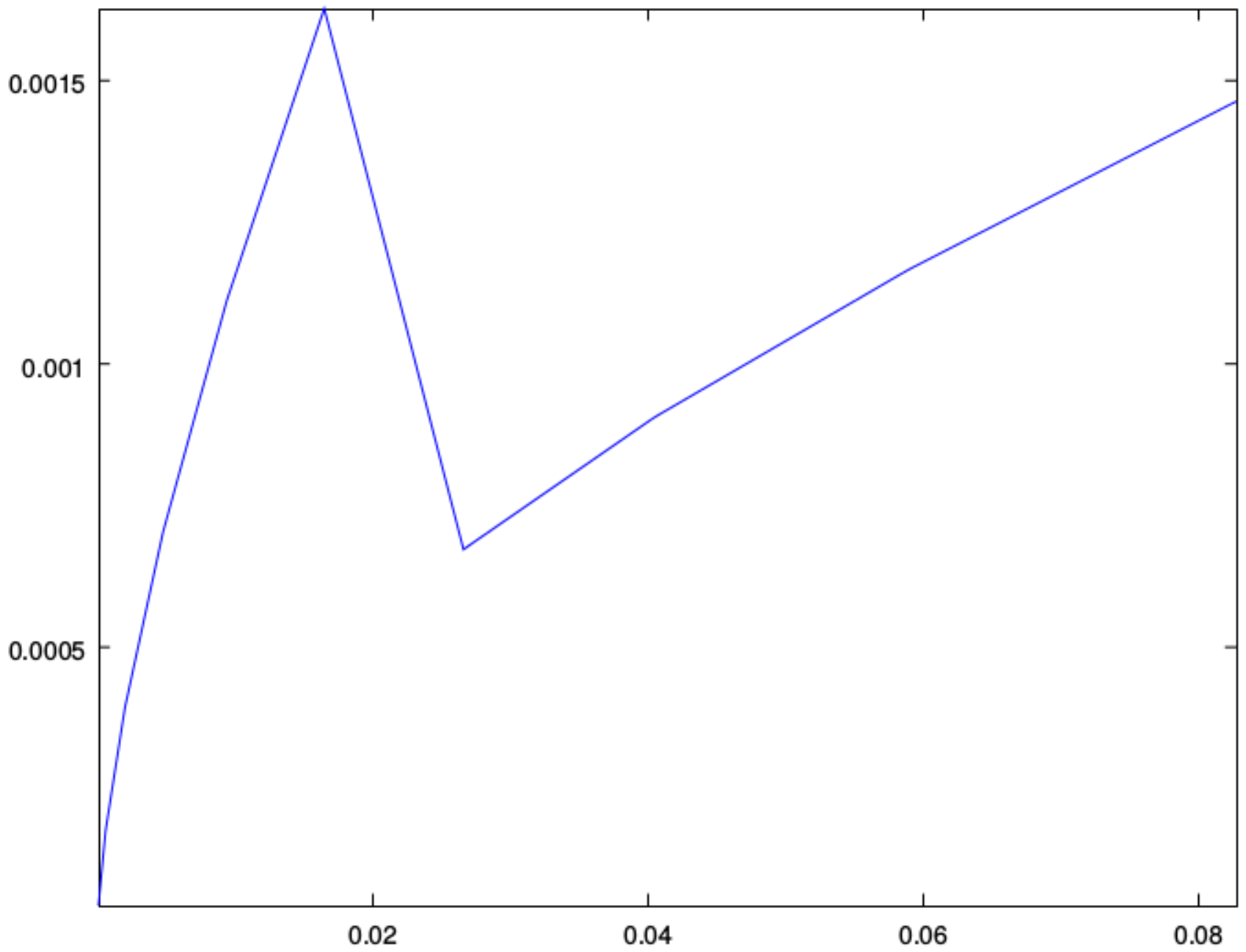}

\centerline{(b)\hspace{6.5cm}(c)}

\caption{Case $\aaa$ for $\eps\cong\Eps$ (a); graph $V\times A$ (b); difference between area values (c).}
\label{aaa}
\end{figure}

Now back to Table~\ref{candid} we consider $\abb$. Numerically speaking this case is rather different from $\aaa$, in which the initial surface was just a tiny equilateral triangle with vertices on the coordinate axes. As usual, {\it Evolver} starts with a very simple polyhedral surface whose triangulation must be consistently refined, equalised and submitted to energy minimisation under geometrical and quantitative constraints. After some iterations the numerical surface may serve as evidence to help answer theoretical questions.

For $\abb$ we could try the initial surface depicted in Fig.~\ref{abb}(a). It consists of a rectangular blue face and a transparent triangular face on top. This one is used by {\tt(q1,q2,q3)} to compute the hyperbolic volume but we ask {\it Evolver} to paint it in {\tt CLEAR} for $A(\Si)$ does not count it.

We end up with a numerical unduloid by starting from Fig.~\ref{abb}(a). But it presents two numerical problems: the convergence is slow and the resulting unduloid is just a local minimum of area under volume constraint. The convergence problem arises from the long thin shape and the initial contact angle of circa $\pi/4$ with the coordinate planes. It does converge to $\pi/2$ as expected from Theorem~\ref{mainthm} but not as quickly as $\aaa$, since now we have that long thin face.

The non-isoperimetric unduloid arises from the metric~(\ref{confm}), which makes the blue face in Fig.~\ref{abb}(a) broad on top and narrow on bottom. That is the feature we get for the resulting surface, which is just a local minimum under volume constraint. From the Euclidean perspective the isoperimetric unduloid must look broader on bottom than on top, which is the case if we start from Fig.~\ref{abb}(b). Notice that the initial surface is also orthogonal to $\deh\B$, which speeds up convergence to a numerical answer compatible with Theorem~\ref{mainthm}.

These strategies may look rather tricky but when dealing with numerical optimisation either {\it Evolver} or any other software will strive for convergence around values that are potentially local minima. Global minima will be hardly found without applying theoretical framework.

That said, we generated unduloids of volume $V(S_\eps)$ according to~(\ref{av}) for $\eps$ varying from $2\Eps/5$ to $6\Eps/5$. This last one is depicted in Fig.~\ref{abb}(c). Finally, Fig.~\ref{abb}(d) shows both $\aaa$ and $\abb$ for $V\times A$. There the $\abb$ and $\aaa$ curves are red and blue, respectively (cf. Fig.~\ref{aaa}(b)). Notice the turning point at $V\cong0.022$.

\begin{figure}[ht!]
\center
\includegraphics[scale=0.27]{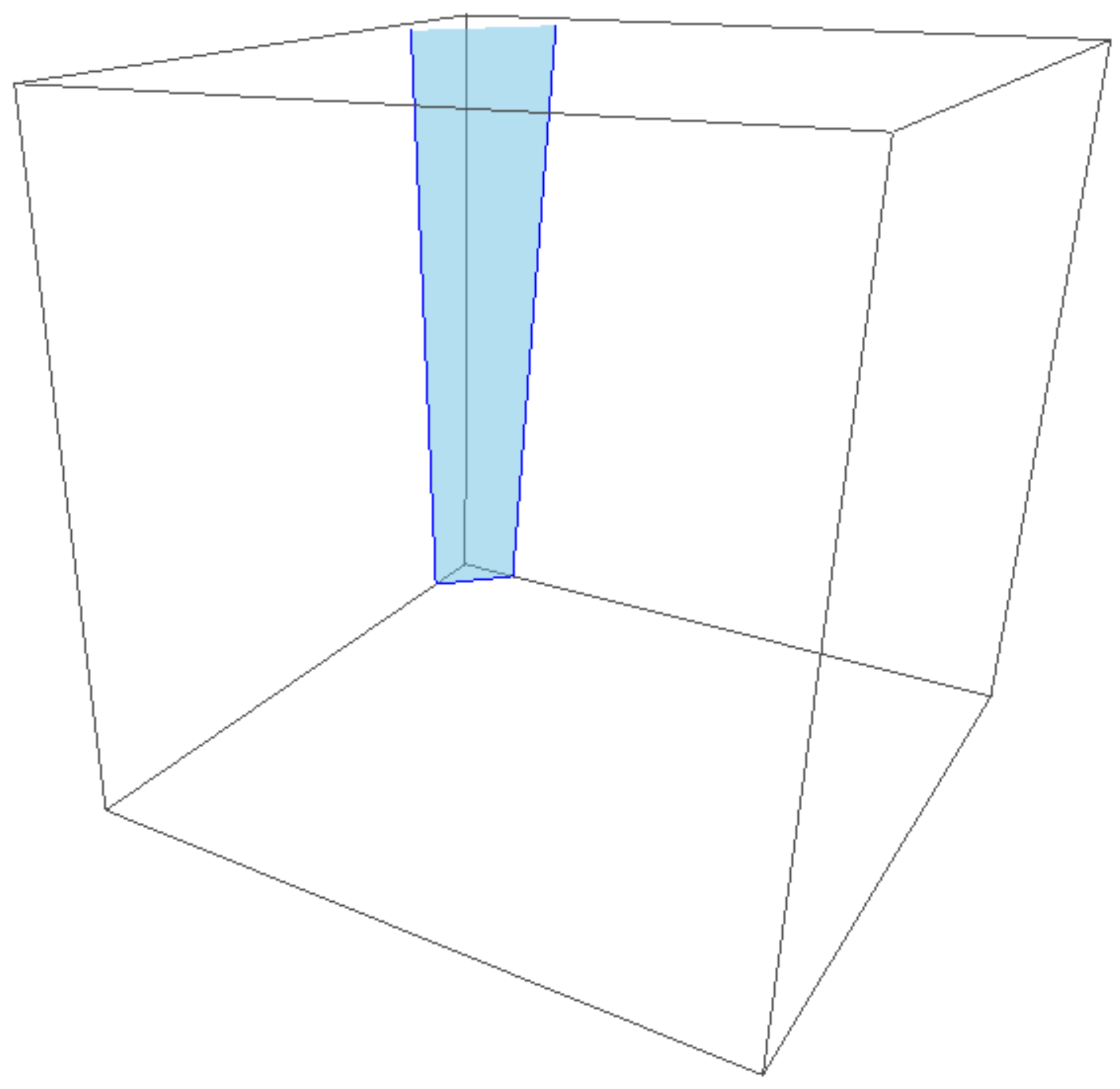}\hfill
\includegraphics[scale=0.30]{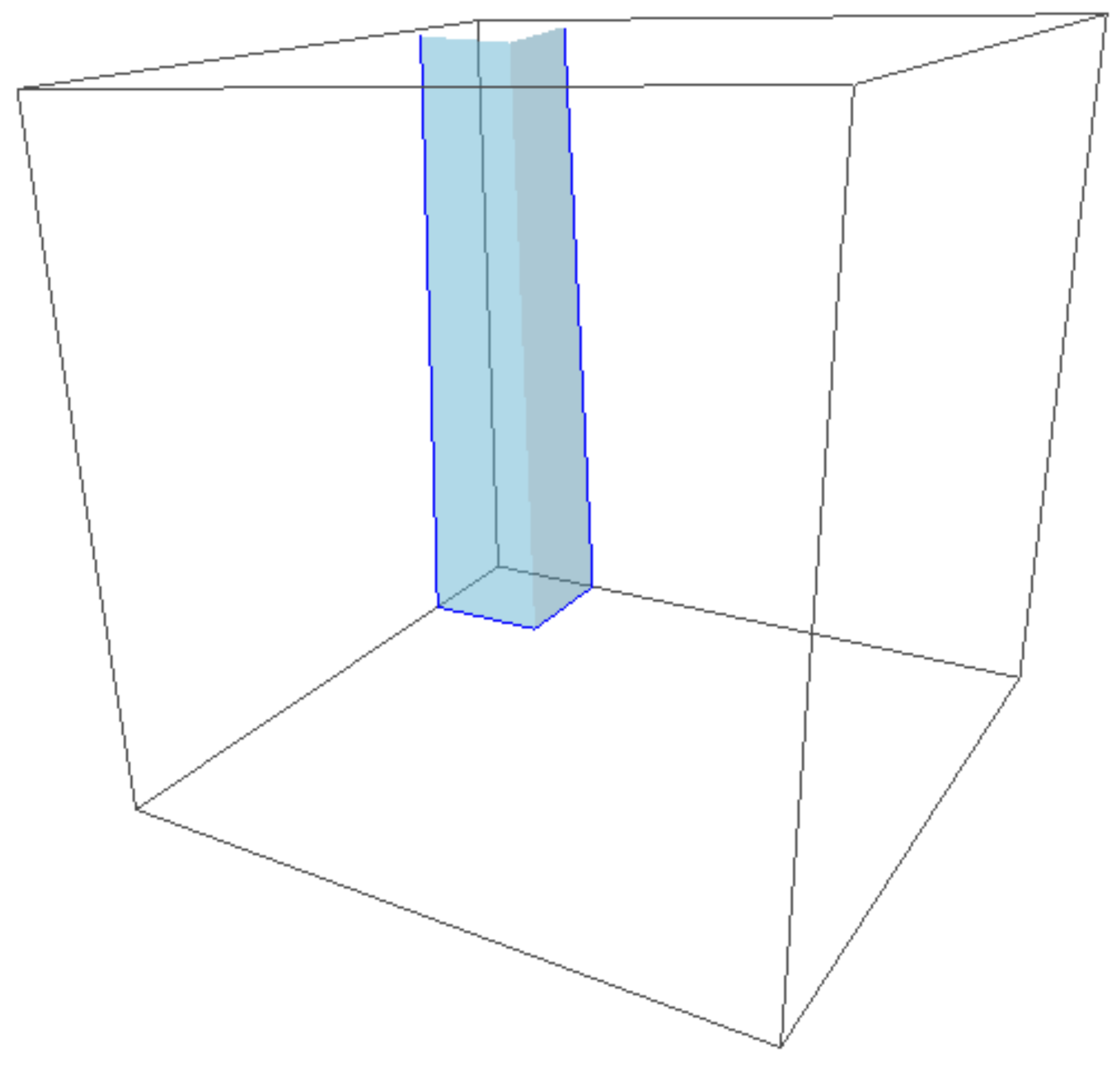}\hfill

\centerline{(a)\hspace{6.5cm}(b)}

\bigskip

\includegraphics[scale=0.27]{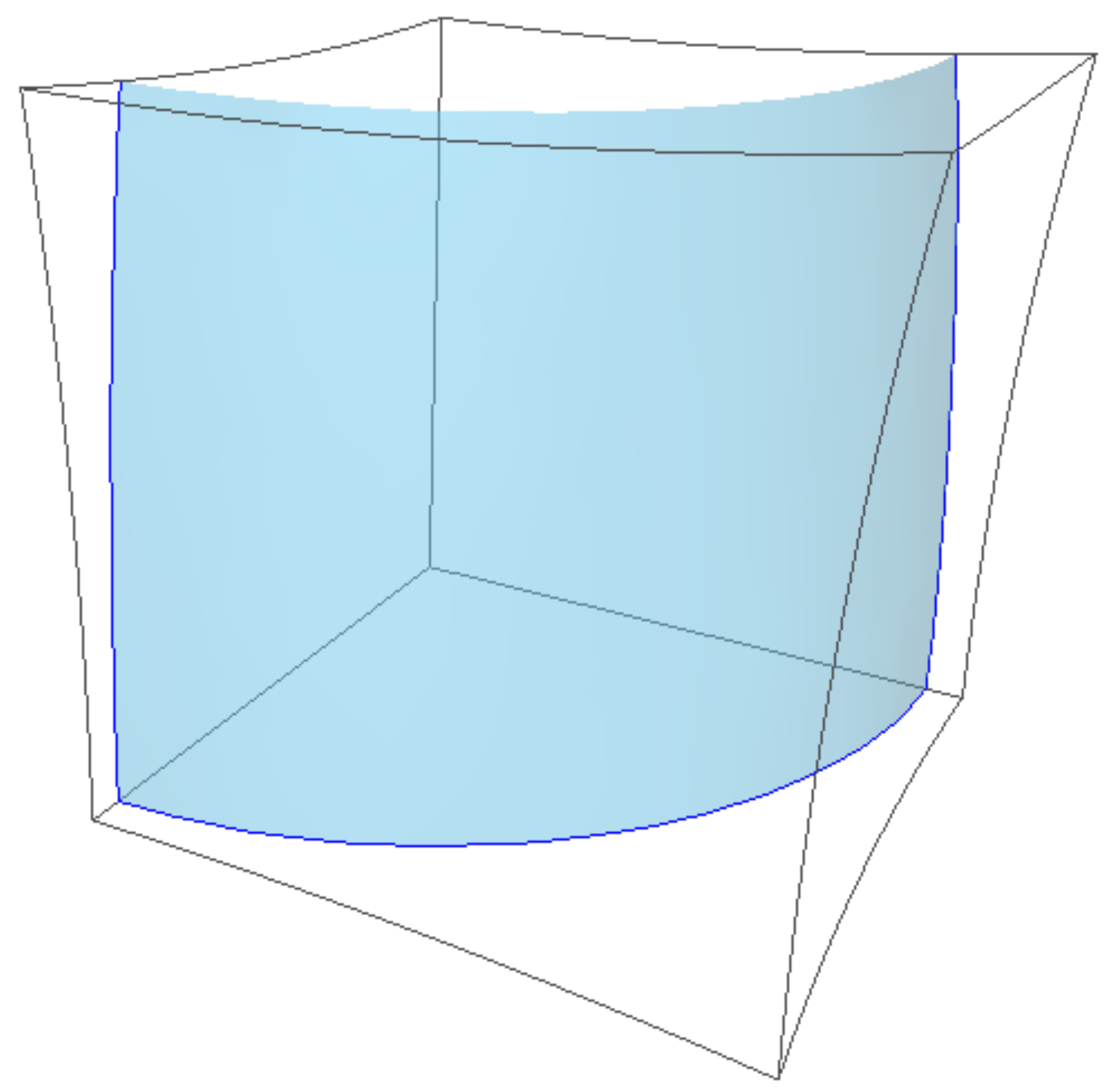}\hfill
\includegraphics[scale=0.32]{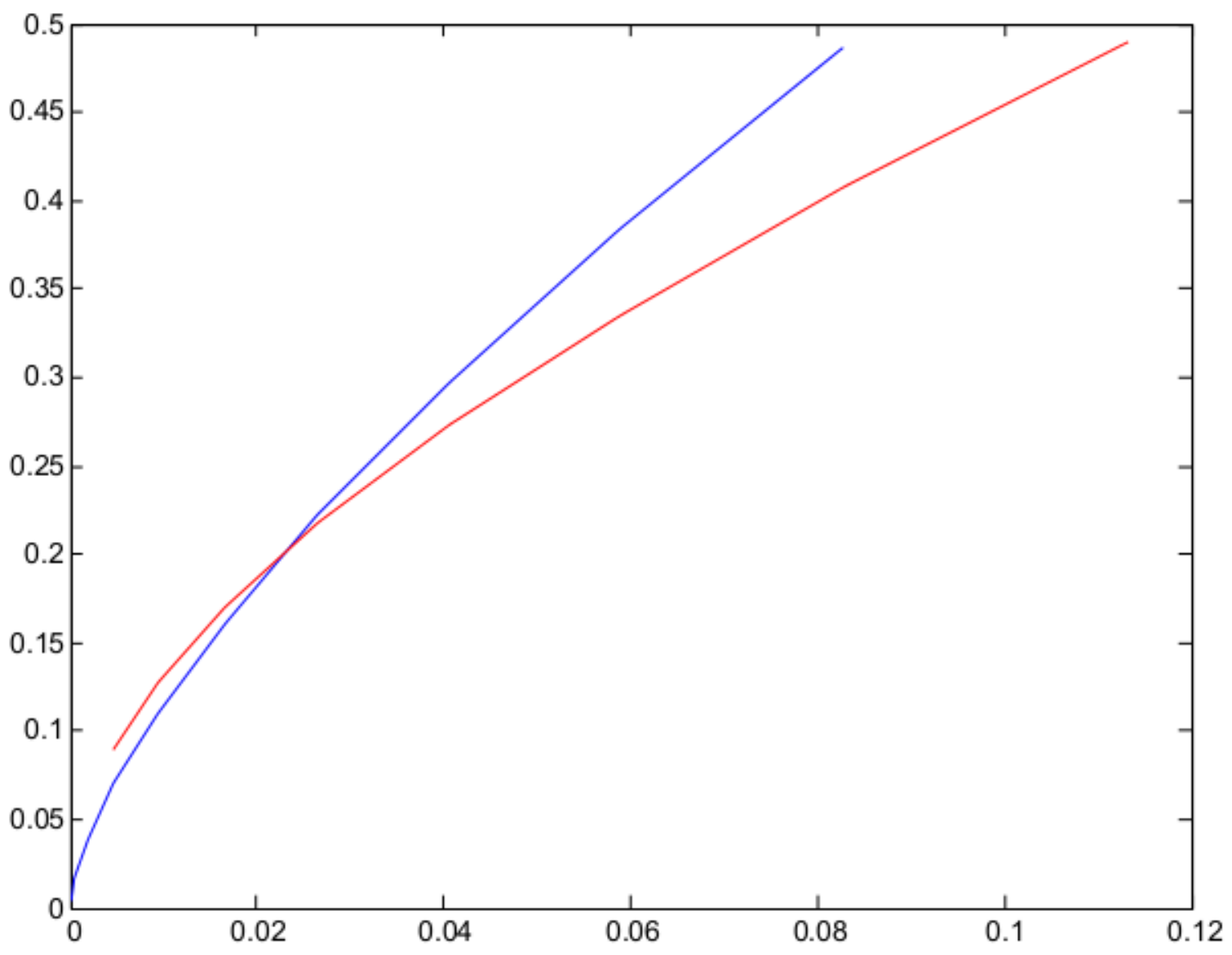}

\centerline{(c)\hspace{6.5cm}(d)}

\caption{Initial $\abb$ (a); with adjustments (b); extreme unduloid (c); $\aaa\times\abb$ (d).}
\label{abb}
\end{figure}

Unduloids like in Fig.~\ref{abb}(c) are surfaces obtained by revolution of a stretch of {\it hypercircle} around $Ox_3$. If we take $Ox_1$, $Ox_3$ as the horizontal and vertical axes in Fig.~\ref{3planes&tiling45}(b), respectively, then hypercircles are all equated as
\BE
   \big(x_1+(1/\eps-\eps)/2\big)^2+x_3^2=(1/\eps+\eps)^2/4,
   \label{horoc}
\EE
in our case restricted to $\eps<\Eps$ and $x_3<c-(r^2-x_1^2)^{1/2}$ in Euclidean values. This motivates the following definition:

\begin{defn}
Consider the curve $(\ref{horoc})$ in $Ox_1x_3$ for positive $x_3<c-(r^2-x_1^2)^{1/2}$ and $\eps<\Eps$. Then $U_\eps$ is the surface of revolution obtained by rotation of this curve around $Ox_3$. We call $U_\eps$ the vertical unduloid with axis $Ox_3$.
\label{ond}
\end{defn}

Back again to Table~\ref{candid} we now take $\bbe$, a much easier case from the numerical point-of-view. For $V=0$ the corresponding $A=\A$ is one quarter of the area of the central square in Fig.~\ref{3planes&tiling45}(b). By means of polar coordinates we get
\BE
   \pi+\A=\int_0^{\pi/4}\frac{2}{1-c^2\cos^2\tet+
   c\cos\tet\cdot\sqrt{c^2\cos^2\tet-1}}=\frac{11}{10}\pi,
   \label{quart}
\EE
where the value in~(\ref{quart}) was computed via {\it Cauchy's Residue Theorem}. Namely, we begin with $A=\pi/10$ and $A$ does not grow very much with $V$. Of course, it is constant in the Euclidean case. In Fig.~\ref{bbe}(a) we see $A\cong0.38$ for the extreme value $V\cong0.15$, which already surpasses $V(\B)/2\cong0.108$. The turning point between $\abb$ and $\bbe$ is at $V\cong0.058$, as shown in Fig.~\ref{bbe}(b). It summarizes the three main cases for our discussions in the next section.

\begin{figure}[ht!]
\center
\includegraphics[scale=0.27]{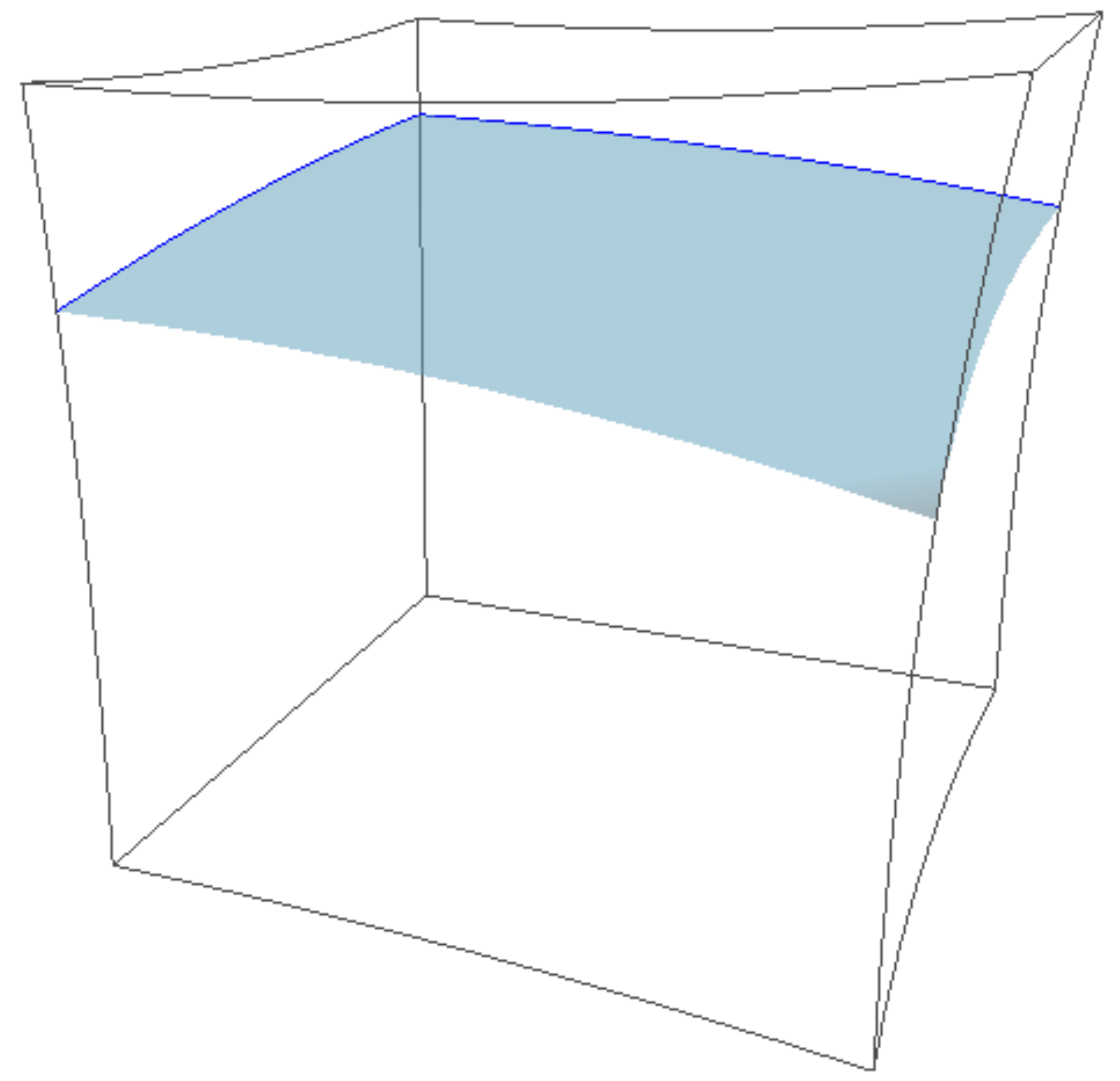}\hfill
\includegraphics[scale=0.32]{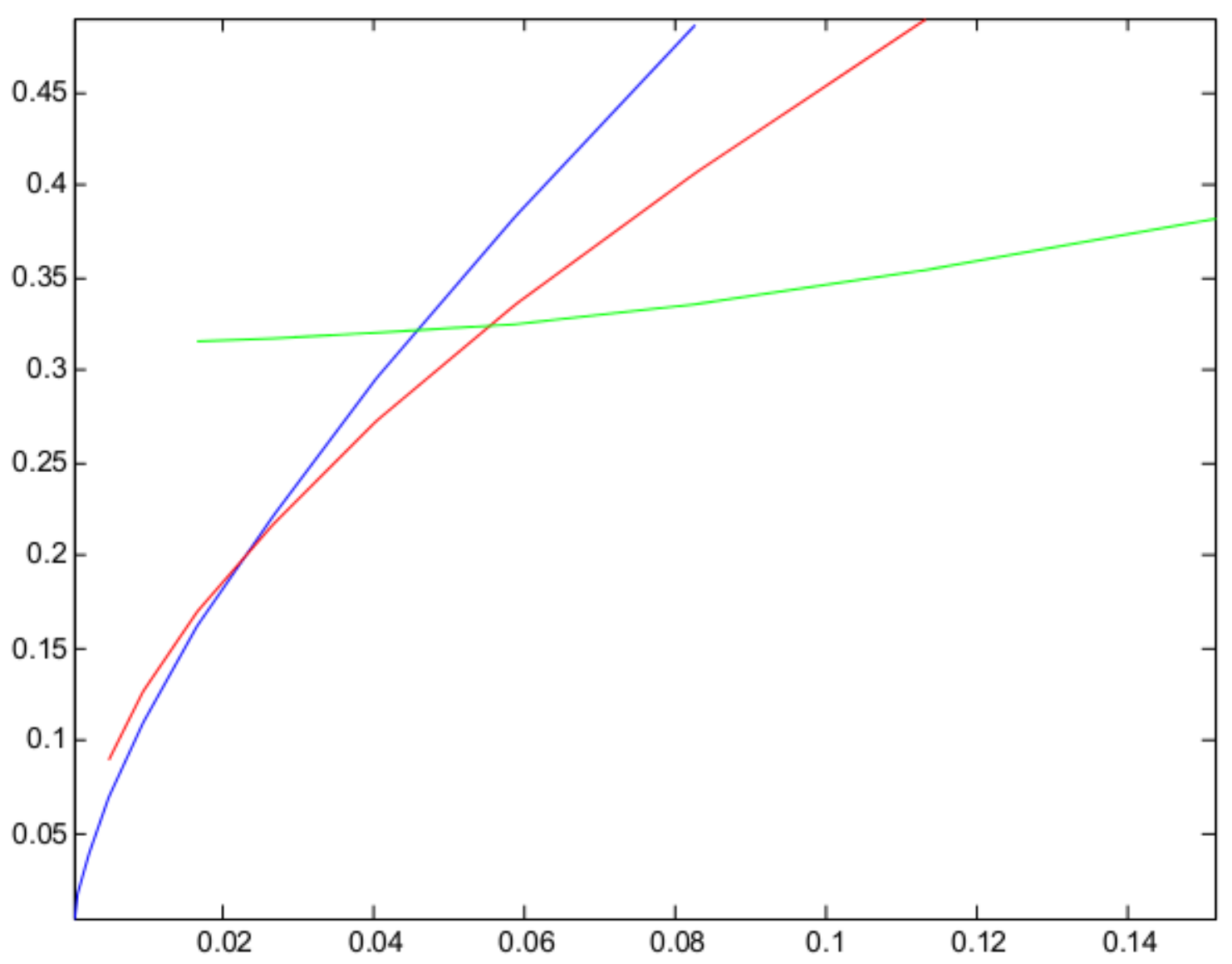}

\centerline{(a)\hspace{6.5cm}(b)}

\caption{Extreme $\bbe$ (a); joint graph of $\aaa$, $\abb$ and $\bbe$.}
\label{bbe}
\end{figure}

Of course, from (\ref{horoc}) we see that the surface in Fig.~\ref{bbe}(a) is a piece of {\it hypersphere} equated as 
\BE
   x_1^2+x_2^2+\big(x_3+(1/\eps-\eps)/2\big)^2=(1/\eps+\eps)^2/4,
   \label{horos}
\EE
where $\eps<\Eps$ and $x_1,\,x_2$ vary inside the positive quadrant of the central square in Fig.~\ref{3planes&tiling45}(b).

\section{Results}
\label{res}

We begin this section by recalling an important question rose in \cite{R1}: Could Lawson surface be a solution of the isoperimetric problem for the cubic lattice in $\R^3$? More precisely, let us take the torus obtained by identifying the opposite faces of the unitary cube in $\R^3$. According to~\cite[p.11]{R1}, with {\it Evolver} one can get a Lawson surface of area 1.017 that encloses the volume $1/\pi$. Technically speaking, if the corresponding theoretical surface has in fact area slightly lesser than 1, then the answer will be {\it yes}. Fig.~\ref{bbd}(a) summarizes the corresponding three cases for a cubic lattice of $\R^3$. Lawson hyperbolic surface is depicted in Fig.~\ref{bbd}(b) with $A\cong0.399$ for $V\cong0.113$, already above $V(\B)/2\cong0.108$. For reasons that will soon be explained the minimum of the three graphs in Fig.~\ref{bbe}(b) will be called {\it the isop-curve}. This curve is compared with the graph $V\times A$ in magenta of our simulation of Lawson's case (see Fig.~\ref{bbd}(c)).

\begin{figure}[ht!]
\center
\includegraphics[scale=0.22]{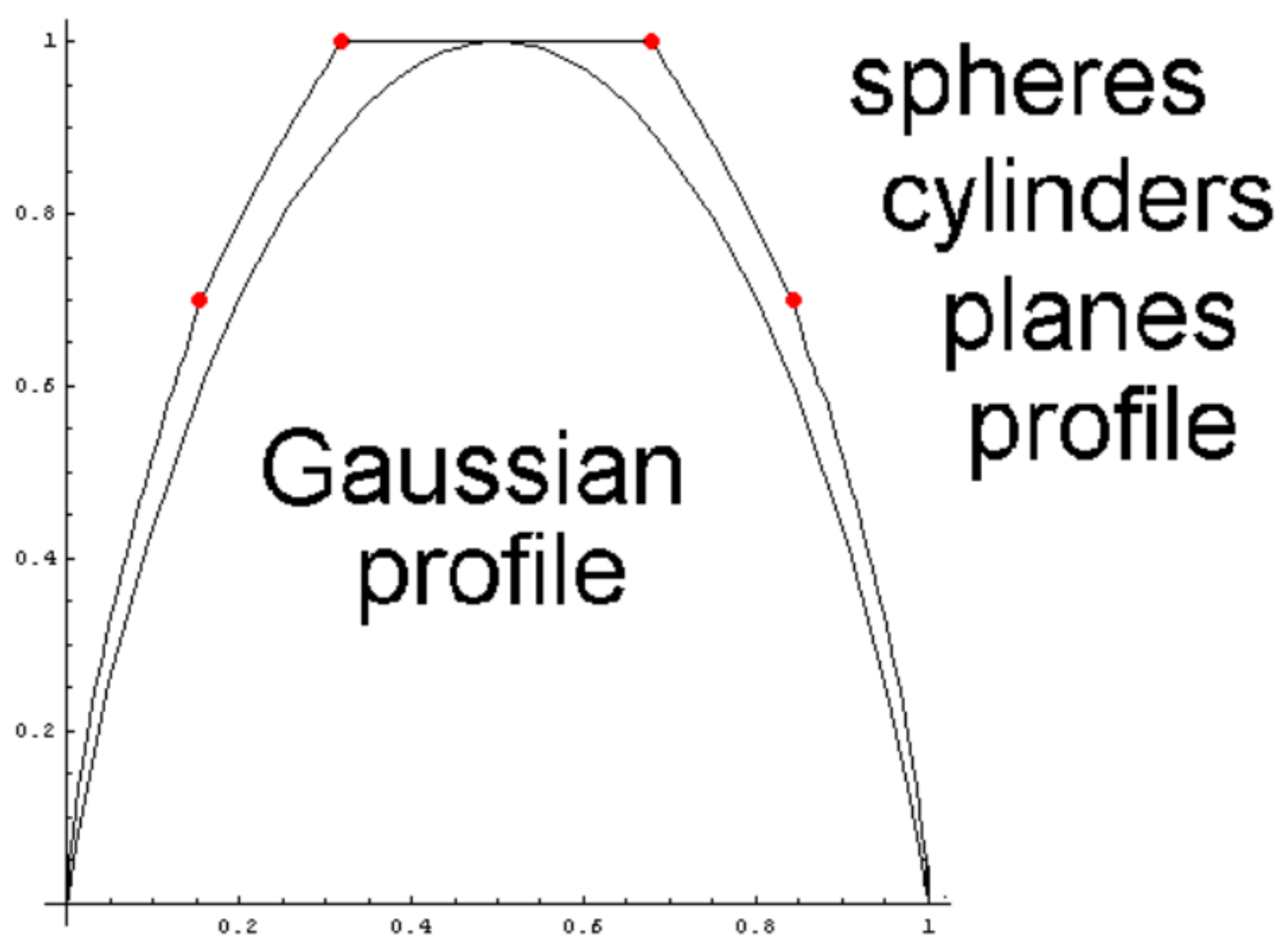}

(a)\\

\smallskip

\includegraphics[scale=0.23]{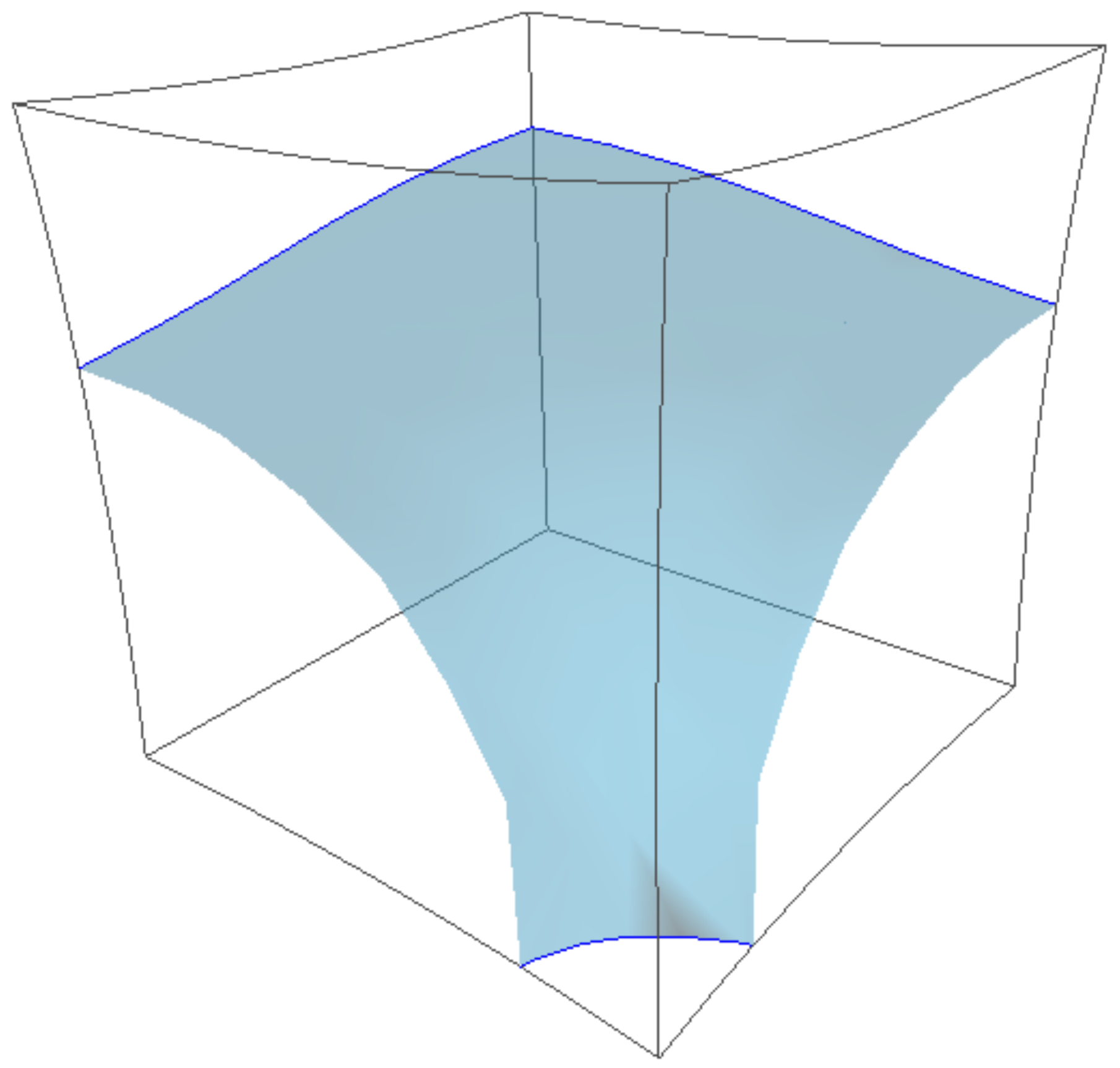}\hfill
\includegraphics[scale=0.30]{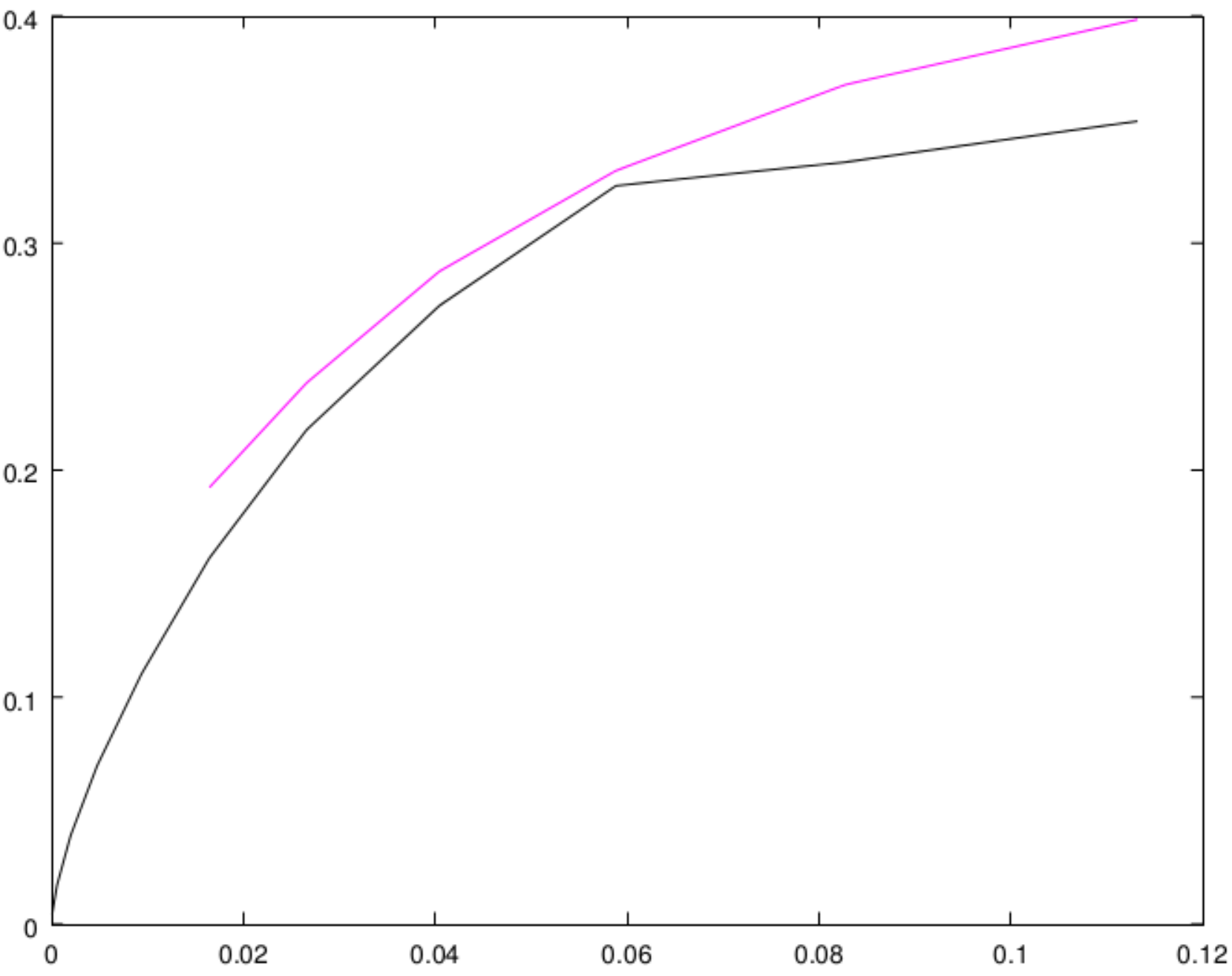}

\centerline{(b)\hspace{6.5cm}(c)}

\caption{The Euclidean case from \cite[Fig.10]{R1} (a); extreme Lawson (b); $\bbd\times$isop-curve in magenta and black, respectively (c).}
\label{bbd}
\end{figure}

In Fig.~\ref{bbd}(a) the turning points occur at circa $0.15$ and $0.3$, respectively. Since $V(\B)\cong0.215$ our first and second turning points occur for the ratios $0.022/V(\B)$ $\cong0.102$ and $0.058/V(\B)\cong0.27$, respectively. Roughly speaking, they happen {\it earlier} than for the Euclidean case. This explains why the $\bbd$-curve is visibly above the isop-curve in Fig.~\ref{bbd}(c).

We used~(\ref{av}) to make the $\bbd$-curve start at $V(S_\eps)\cong0.016$ for $\eps=3\Eps/5$. In Fig.~\ref{abb}(d) the unduloid started at $V(S_\eps)\cong0.005$ for $\eps=2\Eps/5$. In theory an unduloid can have arbitrarily small volume but numerical simulations will only be meaningful if we avoid extreme cases. With Lawson surface there is an additional problem: it simply does not exist when $V$ is too little. In our simulations, for $V$ slightly below $0.016$ ($\eps<3\Eps/5$) the initial surface converges to a degenerate surface consisting of an eighth of $S_\eps$ connected to $(\Eps,0,0)$ and $(0,\Eps,0)$ by tiny tubes.

The analogous problem occurs with $\acc$. Topologically speaking, this case is obtained through a Euclidean reflection of \bbd\, in the plane determined by the points $(\Eps,0,0)$, $(0,\Eps,0)$, $(\cfk,0,\cfk)$ and $(0,\cfk,\cfk)$, where $\cfk$ is defined in~(\ref{corn}). This plane works as a mirror that inverts the way we see Lawson surface, hence the name {\it inverted Lawson}.

In the Euclidean cube of edgelength $\cfk$ the inverted Lawson is again Lawson surface but translated by $(\cfk,\cfk,0)$. However, in the hyperbolic cube we get a different graph $V\times A$ for the inverted surface. This is due to the metric~(\ref{confm}), which requires more area $A(\Om)$ to comprise the same volume $V(\Om)$ when $\Om$ leaves the origin (see Example~\ref{exsph} in Sect.~\ref{prelim}). Fig.~\ref{acc}(a) shows an inverted Lawson close to the extreme case of collapsing its handle around $O$ with the origin itself. Fig.~\ref{acc}(b) shows the $V\times A$ graph for $\acc$ in blue, together with the other graphs in Fig.~\ref{bbd}(c). 

\begin{figure}[ht!]
\center
\includegraphics[scale=0.23]{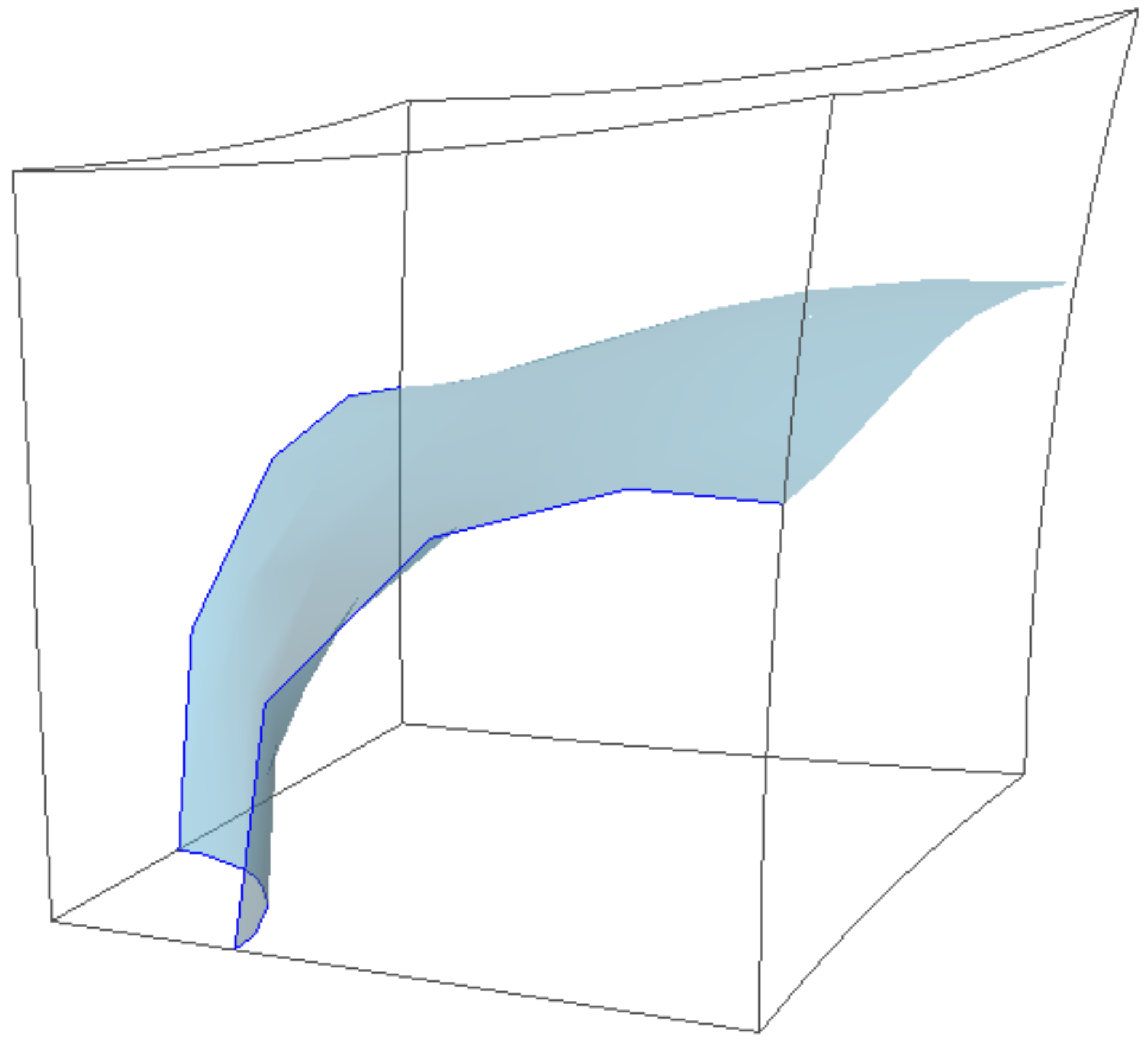}\hfill
\includegraphics[scale=0.30]{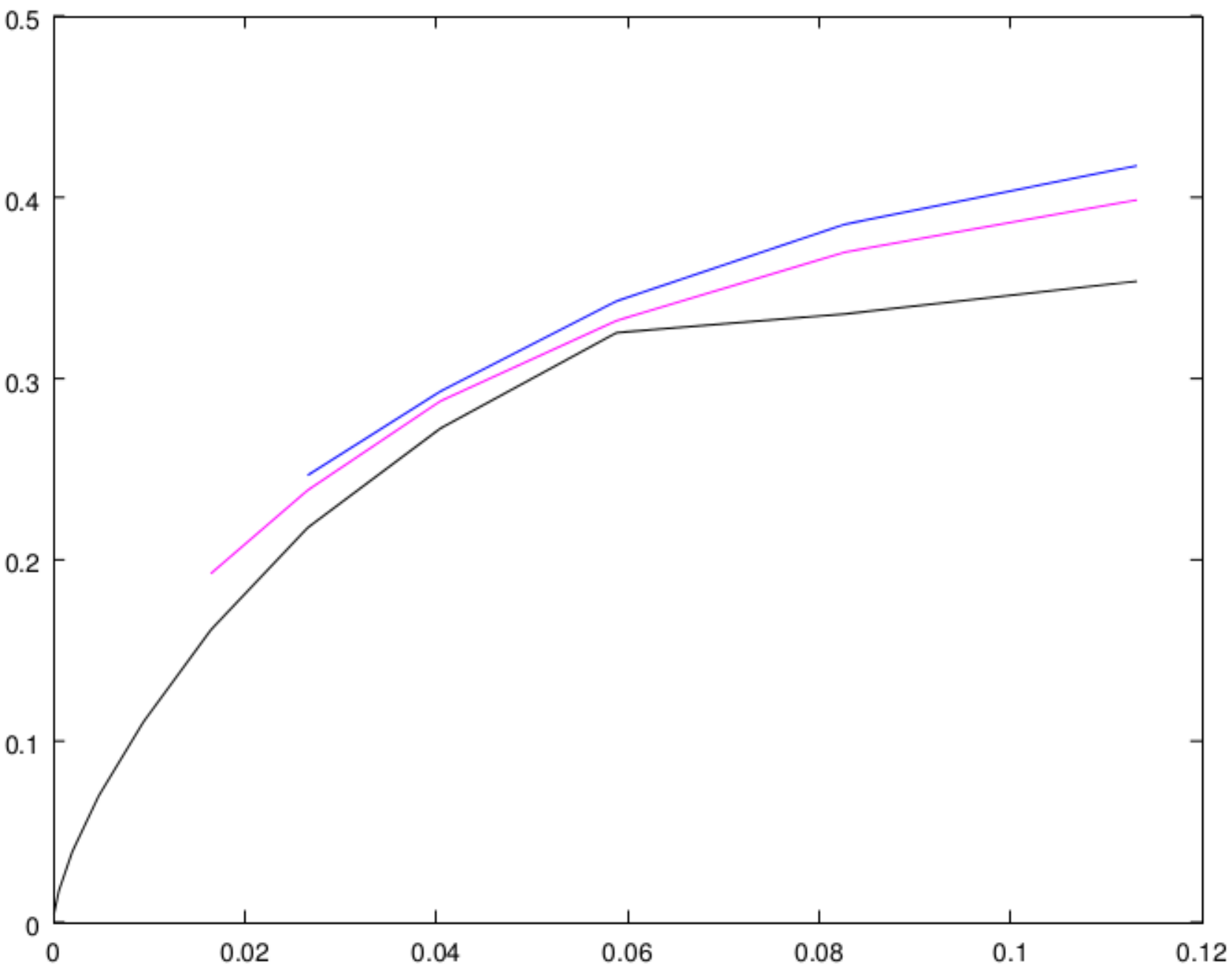}\hfill

\centerline{(a)\hspace{6.5cm}(b)}

\caption{Extreme inverted Lawson (a); $\acc$ in blue, $\bbd$ and the isop-curve (b).}
\label{acc}
\end{figure}

Since we have just recalled Example~\ref{exsph}, notice that its arguments also work for unduloids:

\begin{expl}
Take $U_\eps,\,U$ instead of $S_\eps,\,S$, respectively, where $U_\eps$ is given by Definition~\ref{ond} and $U$ is the ``inverted $U$'', as explained right above. Namely, $V(U)=5V(U_\eps)/4$ and therefore $A(U\cap\B)>A(U_\epsilon\cap\B)$. In the case of a horizontal unduloid, namely with axis either $Ox_1$ or $Ox_2$, area and volume remain unchanged for the ``inverted'' surface. 
\label{exund}
\end{expl}

In Example~\ref{exund} the assertion about horizontal unduloids comes from the fact that their ``inversion'' can be equated. For the axis $Ox_2$ it is the composition of two isometries in $\H^3$: $90^\circ$-rotation $(x_1,x_2,x_3)\mapsto(x_2,-x_1,x_3)$ followed by the spherical inversion
\[
   q\mapsto2r(c+r)\cdot\frac{q-(c+r,0,0)}{||q-(c+r,0,0)||^2}+(c+r,0,0).
\]
For the axis $Ox_1$ it is $(x_1,x_2,x_3)\mapsto(-x_2,x_1,x_3)$ followed by
\BE
   q\mapsto2r(c+r)\cdot\frac{q-(0,c+r,0)}{||q-(0,c+r,0)||^2}+(0,c+r,0).
   \label{invy}
\EE

Because of Examples~\ref{exsph} and~\ref{exund} we had already expected the $\acc$-graph to lie above the $\bbd$-graph as depicted in Fig.~\ref{acc}(b). For the sake of concision we could then have skipped the numerical analysis of $\acc$. Nevertheless, it is important to include it here for two reasons: it justifies why Table~\ref{candid} omits even some feasible cases (like $\ccd$ and $\dde$), and it helps check the reliability of our numerical simulations as we are going to do right now.

Indeed, $\acc$ above $\bbd$ was already expected but we can look even closer. The factor $5/4$ is due to the five-fold symmetry on the corners instead of the four-fold symmetry at the centre of any square in Fig.~\ref{3planes&tiling45}(b). In the Euclidean case any 3D-manifold whose dimensions increase by the same factor $\ld$ will get $A,\,V$ increased by $\ld^2$ and $\ld^3$, respectively. For a cylinder of constant height we have the same factor $\ld$ for both $A,\,V$ if top and bottom do not count. Roughly speaking Lawson surface is close to a pair of horizontal unduloids connected by a very small piece of sphere. By means of Taylor expansion we rewrite (\ref{av}) as
\BE
   A(S_\eps)=16\pi\eps^2+\O(\eps^4)
   \hspace{1cm}{\rm and}\hspace{1cm}
   V(S_\eps)=\frac{32}{3}\pi\eps^3+\O(\eps^5).
   \label{tay}
\EE
Namely $A$ grows by $(5/4)^{2/3}\cong1.16$ when $V$ grows by $5/4$, but its contribution to Lawson surface must be lesser than that. As we have already explained in Example~\ref{exund}, the pair $(V,A)$ remains unchanged for horizontal unduloids. Therefore $A$ must grow very little for the inverted Lawson, as we see in Fig.~\ref{acc}(b). The growth ratio varies from $1.03$ to $1.05$ according to our numerical tests.

A little reflection shows that $\bce$ is the ``inversion'' of {\tt bdb}, this one congruent to Lawson surface, and now we have arguments to skip $\bce$. Table~\ref{candid} also omits $\ccd$ because it is the ``inverted'' $\bcd$. One obtains $\dde$ by taking the ``inversion'' of $\abc$ with respect to the rectangle $(\eps,0,0)$, $(0,0,\eps)$, $(\cfk,\cfk,0)$, $(0,\cfk,\cfk)$, hence this one was omitted from Table~\ref{candid} as well.

Now we are going to study $\ddd$, namely Schwarz surface in Table~\ref{candid}. For reasons that we have already explained for the case of Lawson surface, it is already expected that Schwarz surface does not exist for too little volume. Indeed, our numerical experiments show that the initial surface degenerates to a piece of sphere connected to $(\Eps,0,0)$, $(0,\Eps,0)$ and $(0,0,\Eps)$ by tiny tubes when $V$ is too little. We begin with $V_\eps$ for $\eps=4\Eps/5$ until $\eps=11\Eps/10$ when the volume surpasses $V(\B)/2\cong0.108$. See Fig.~\ref{ddd}(a) for an illustration and also Fig.~\ref{ddd}(b) for its graph $V\times A$ in red.

\begin{figure}[ht!]
\center
\includegraphics[scale=0.21]{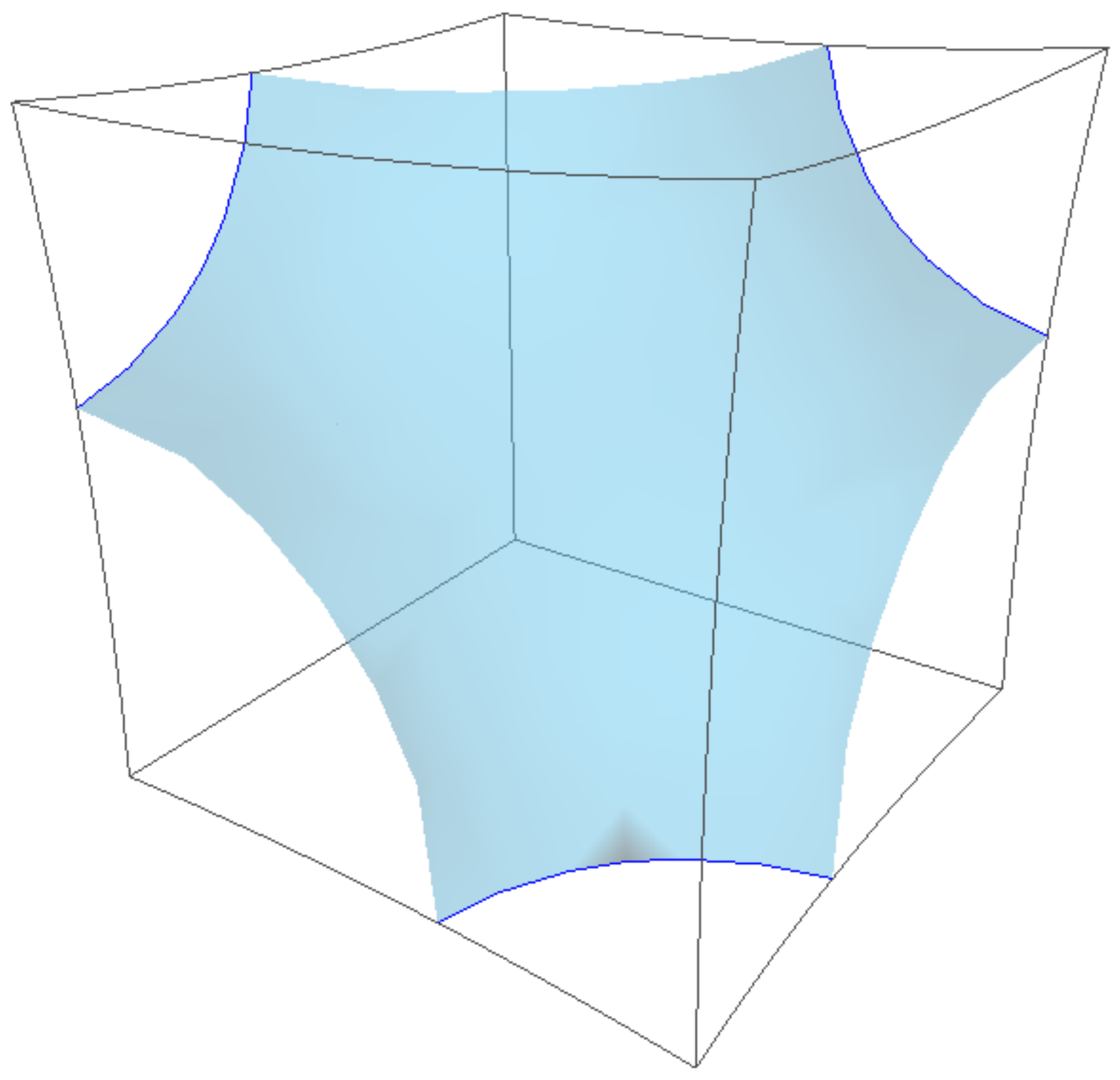}\hfill
\includegraphics[scale=0.30]{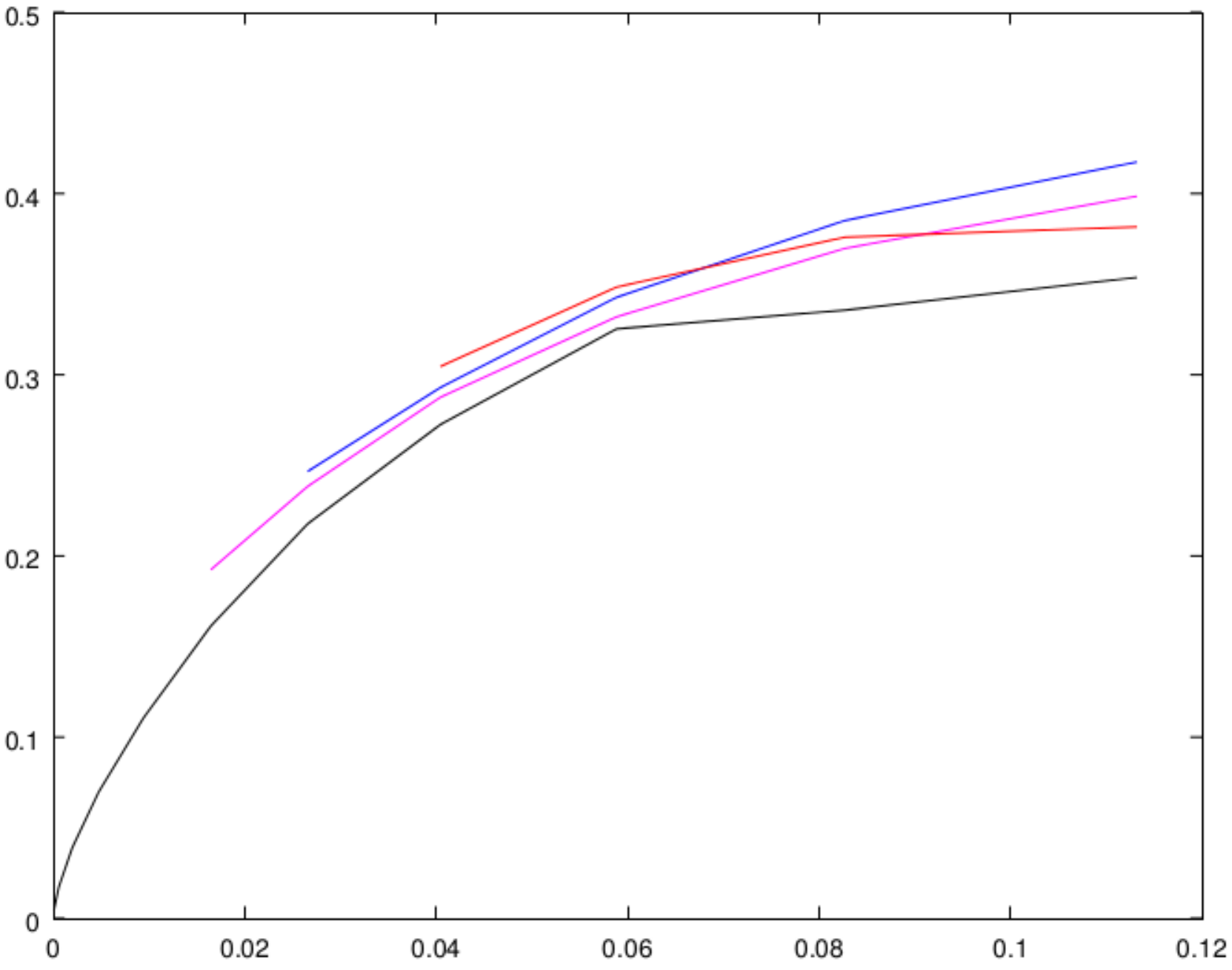}\hfill

\centerline{(a)\hspace{6.5cm}(b)}

\caption{Extreme Schwarz (a); $\ddd$ in red and the previous curves (b).}
\label{ddd}
\end{figure}

For the readers who want to have an overview of Schwarz's surface with zero CMC in~$\H^3$ here is a nice picture by Konrad Polthier (Freie Universit\"at Berlin): \url{http://www.polthier.info/articles/teaser/H3wrfeb_tiny.jpg}

Now what is the {\it inverted Schwarz}? By looking at Figs.~\ref{cases}(a) and (c) we realise that it should be again $\acc$, which was called {\it inverted Lawson}. However, by arguments already given in this section we know that ``inversion'' leads to a non-isoperimetric $\Si$ and for this reason we do not care about the apparent dubious meaning of $\acc$. The {\it inverted Schwarz} can then be skipped, even if Fig.~\ref{cases} leaves only $\acc$ as a way to codify it. For the same reason we shall not differ sub-cases when dubious triads $\ell_1\ell_2\ell_3$ appear again.

Back to Table~\ref{candid} we now study the remaining cases $\abc$, $\bcd$ and $\bce$. Topologically speaking Fig.~\ref{abc}(a) shows that these surfaces are Lawson's example reflected by (\ref{invy}). But differently from the previous inverted cases now we must analyse these ones numerically. Fig.~\ref{abc}(b) compares $\abc$ with Fig~\ref{bbd}(c).

\begin{figure}[ht!]
\center
\includegraphics[scale=0.25]{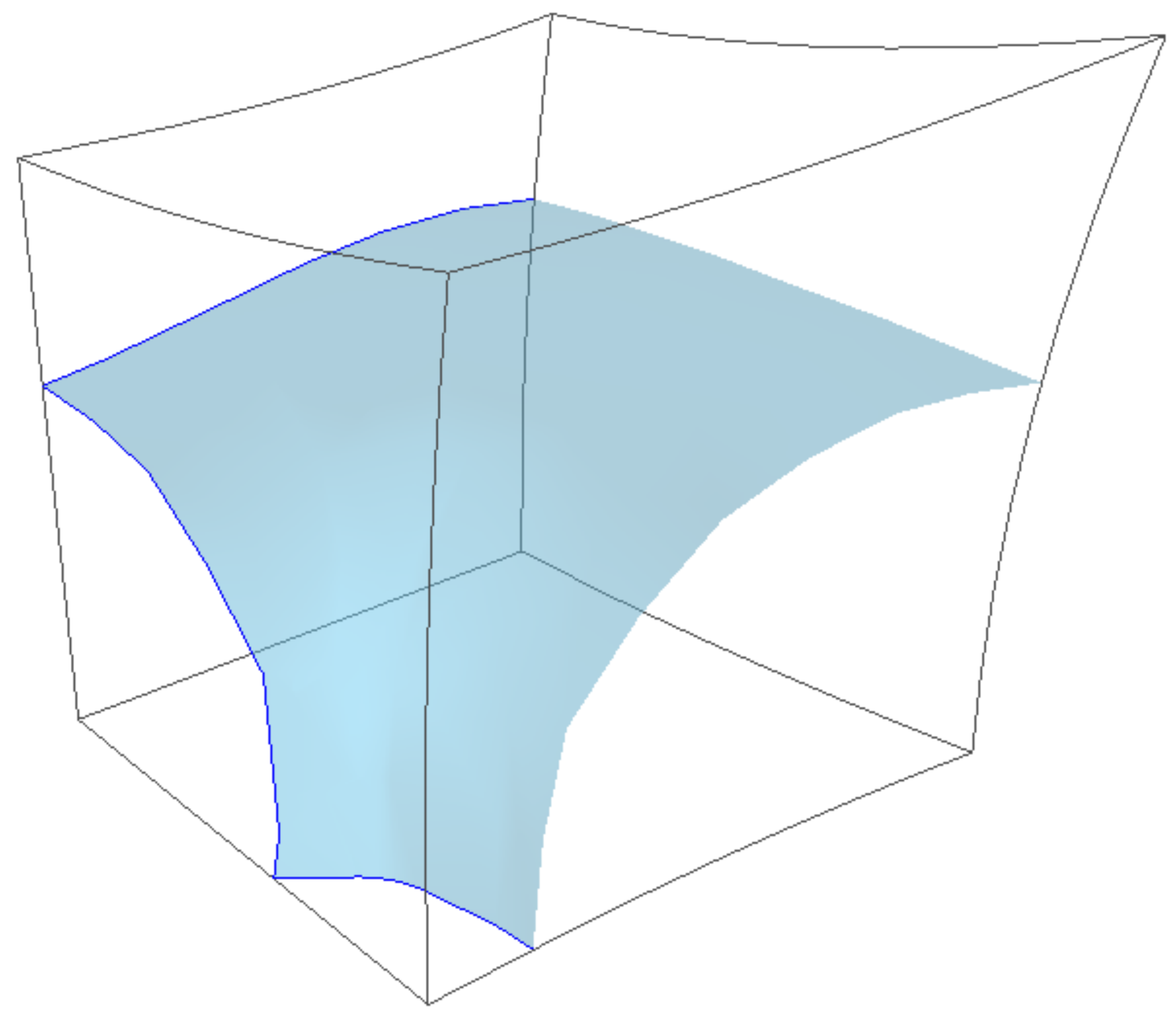}\hfill
\includegraphics[scale=0.40]{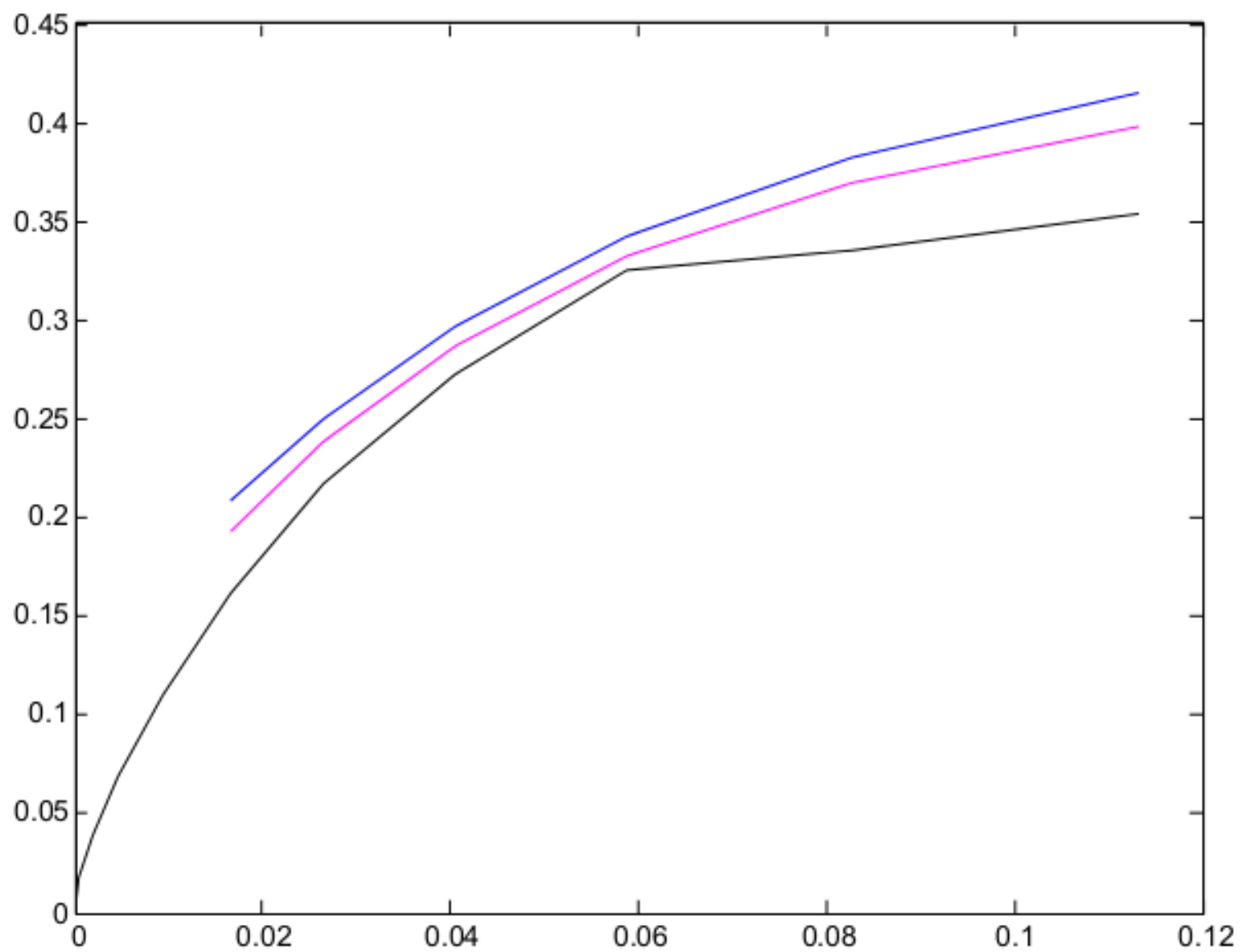}\hfill

\centerline{(a)\hspace{6.5cm}(b)}

\caption{Extreme $\abc$ (a); $\abc$ in blue, $\bbd$ and the isop-curve (b).}
\label{abc}
\end{figure}

By looking at Figs.~\ref{abc}(b) and~\ref{acc}(b) we realise that $\abc$ is quite similar to the inverted Lawson. But $\abc$ lies above $\bbd$ almost as a parallel curve while $\acc$ rises more quickly and is closer to $\bbd$ for small values of $V$. Anyway, we still can argue that $\abc$ needs more area for the same volume of $\bbd$ because the contact of $\abc$ with the origin is not as ``large'' as in the case of $\bbd$. Together with the metric~(\ref{confm}) this explains why $\abc$ performs worse than $\bbd$. 

Finally we are going to study $\bcd$, which is in fact a degenerate case. From Figs.~\ref{bcd}(a) to (c) one promptly recognises that its topology was discarded as a possible CMC surface in a Euclidean three-dimensional box, as proved in~\cite{Ri}. However, since we cannot adapt all of the arguments in~\cite{Ri} to the hyperbolic geometry, and similar CMC surfaces were already found in Euclidean boxes (see the gyroids of \cite[Fig.2]{KGB}), then we must include the numerical analysis of this case. The volume $V\cong0.058$ is fixed and the area $A$ decreases in Figs.~\ref{bcd}(a) to (c). Namely, no CMC surface can be found in $\B$ with that topology.

\begin{figure}[ht!]
\center
\includegraphics[scale=0.18]{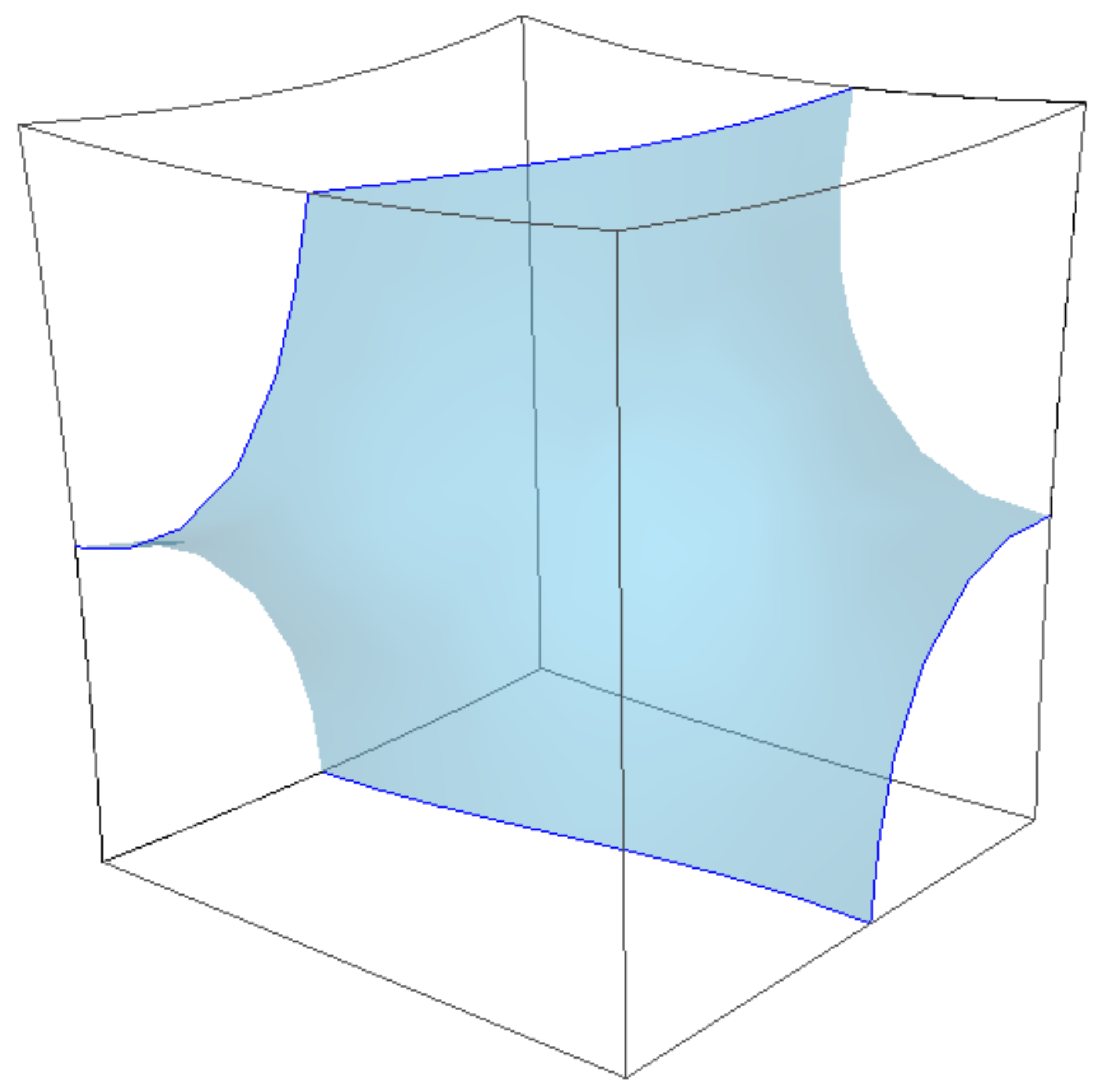}\hfill
\includegraphics[scale=0.18]{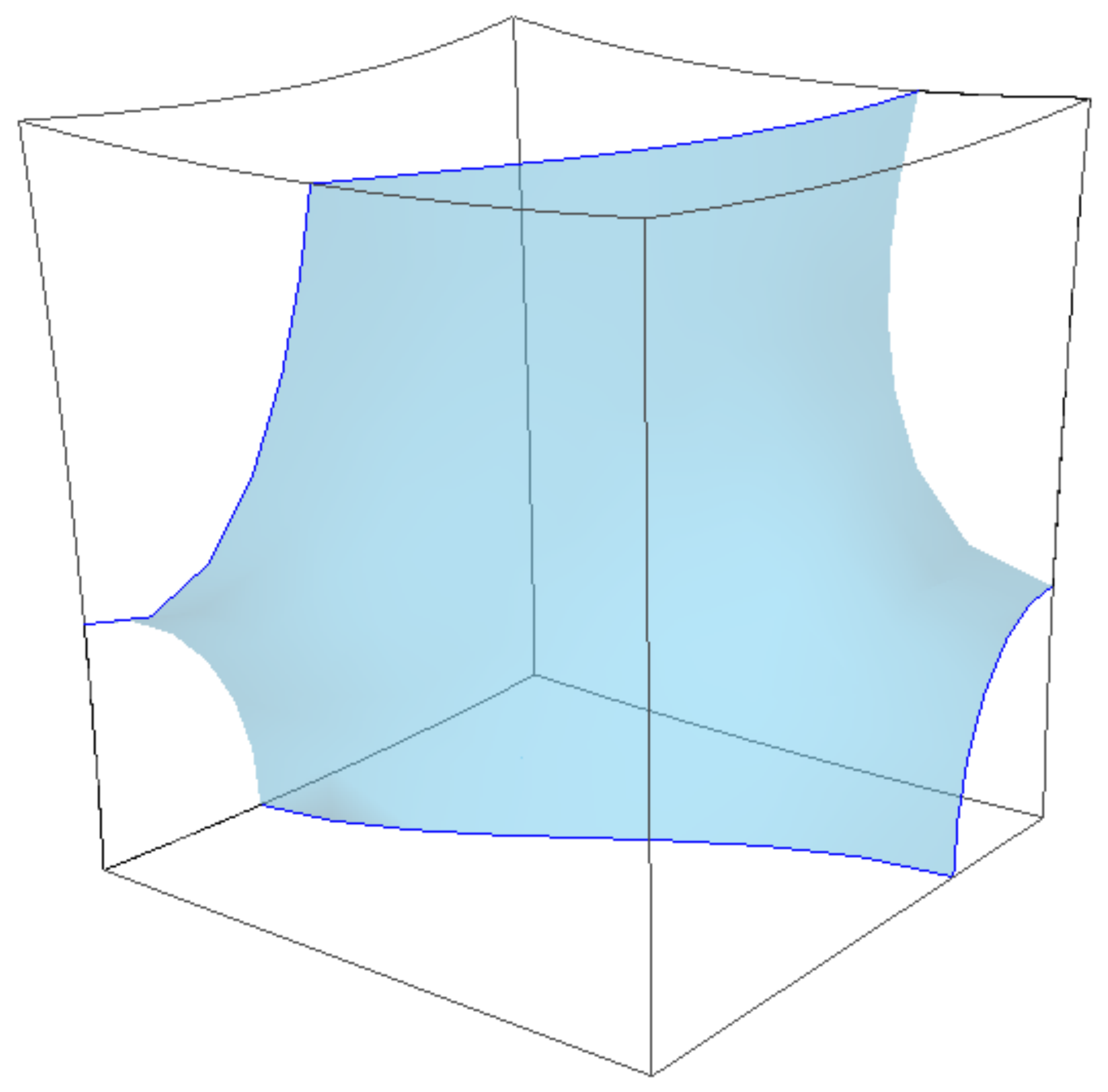}\hfill
\includegraphics[scale=0.18]{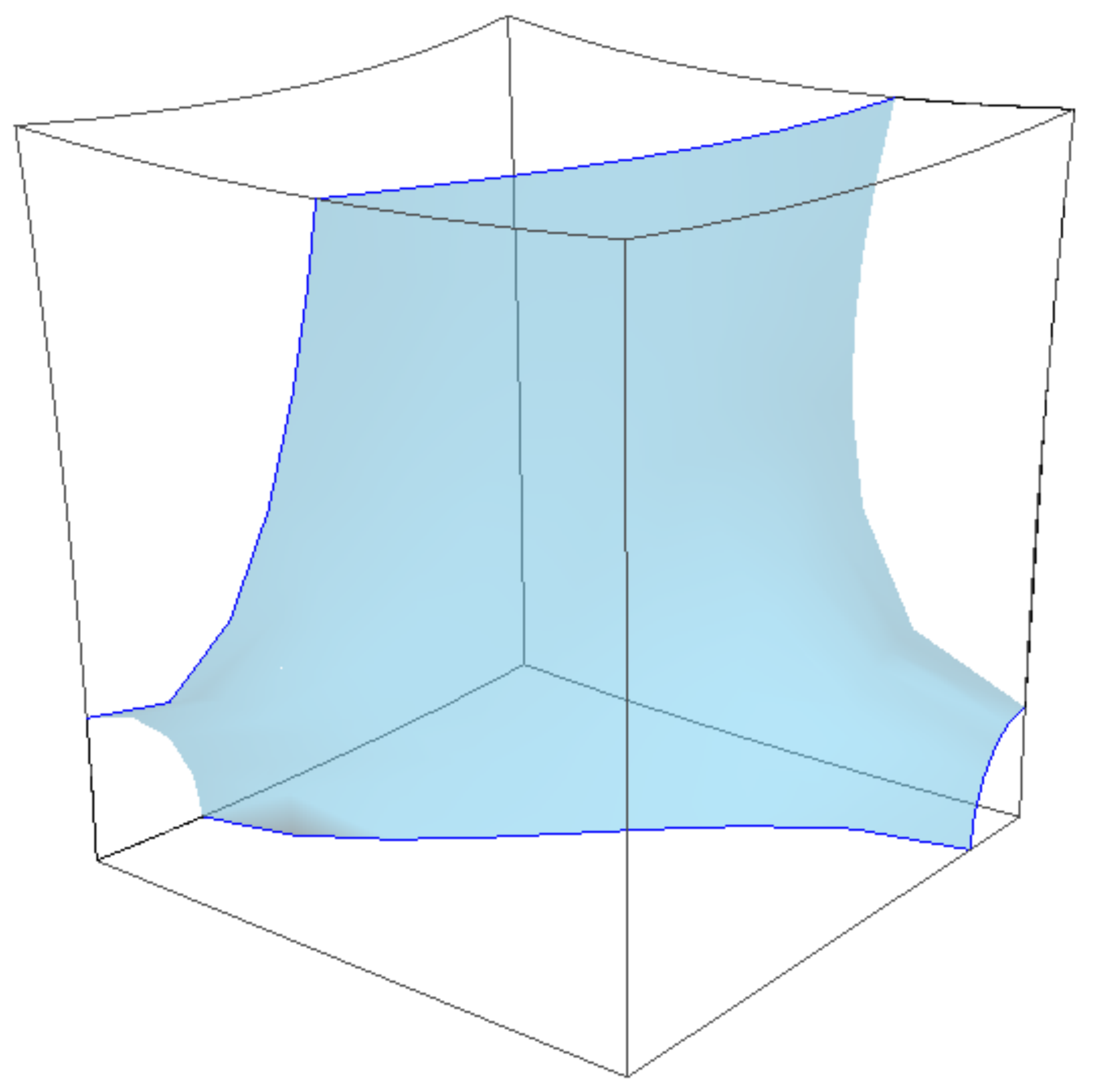}

\centerline{(a)\hfil(b)\hfil(c)}

\caption{The degenerate case $\bcd$: initial smoothed surface (a); evolving once (b); close to degeneration (c).}
\label{bcd}
\end{figure}

We conclude this section with a brief comment on the cases depicted in Fig.~\ref{extra}. One of them is combined with {\tt aa} and presented in Fig.~\ref{dsc}. It looks like a Lawson surface rotated by $45^\circ$ around $Ox_3$. What is a congruence in the Euclidean case changes drastically in $\B$. Fig.~\ref{dsc} shows an example with $V\cong0.04$ and $A\cong0.328$, namely well above the corresponding $V$ for the magenta curve in Fig.~\ref{bbd}(c). The magenta curve starts at $V\cong0.016$ but now what happens for $V$ slightly below $0.04$ is that the surface degenerates to a pair of eighths of sphere centred at $O$ and $(\cfk,\cfk,0)$ and connected by a tiny tube.

\begin{figure}[ht!]
\center
\includegraphics[scale=0.35]{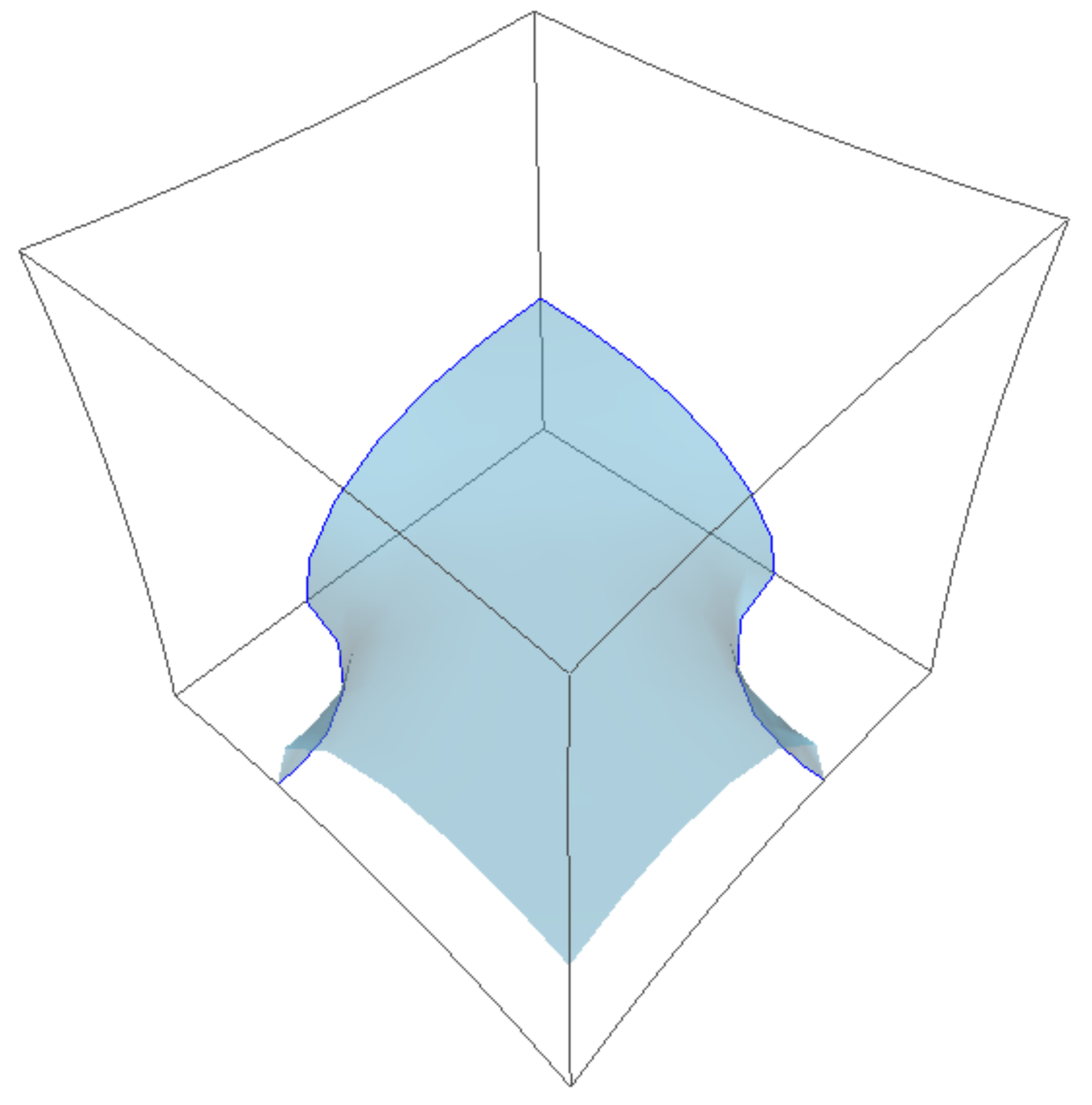}

\caption{The case {\tt aa} plus a discontinuous graph.}
\label{dsc}
\end{figure}

A little reflection explains the bad performance of the cases in Fig.~\ref{extra}: they make $\Om$ have less contact around $O$ with the fundamental planes compared with any option in Table~\ref{candid}. Indeed, the example in Fig.~\ref{dsc} is above all others in that table if we look at $V=0.04$ in Figs.~\ref{ddd}(b) and~\ref{abc}(b). That is why we have not included the cases of Fig.~\ref{extra} in Table~\ref{candid}.

Let us now resume our discussion at the beginning of this section. Lawson's result can be found in~\cite{L} but for convenience of the reader we shall reproduce it here:

\begin{thm} There exist two doubly periodic surfaces of constant mean curvature one and genus two contained in a slab of $\R^3$. The ambient translational fundamental cell is a hexagonal or square prism (of infinite height), respectively.
\label{lawson}
\end{thm}

Namely, we had been focusing our attention to just one of Lawson's examples because our study is devoted to the cubic lattice. It is important to mention, however, that in~\cite[p358]{R2} the author comments on some computer simulations for the isoperimetric problem in a skewed Euclidean box. These were performed with the Surface Evolver by Pascal Romon (Universit\'e Paris-Est Marne-la-Vall\'ee), who has found an isoperimetric surface of genus two~\cite[Fig.3]{R2}, precisely the hexagonal example of Theorem~\ref{lawson}. See~\cite{R2} for details.

\section{Conclusions}
\label{concl}

At the beginning of Sect.~\ref{res} we recalled a question rose in~\cite{R1}. It remains open until the present day but the corresponding question in $\B$ has great chance to be answered because of the evidence shown in Fig.~\ref{bbd}(c). There the $\bbd$-curve is at least 10\% distant from the isop-curve, whereas in the Euclidean case it is only 1.7\%. This is consistent with the fact that close objects in the Euclidean viewpoint turn out to be farther apart in the hyperbolic geometry. Curiously in both cases the closest approach happens at the turning point between the two surfaces of genus one.

The numerical gap of 10\% makes it easier to prove that no Lawson surface can be isoperimetric in a hyperbolic space form: the fist step it to simplify the elliptic equations that define these surfaces by means of a theoretical and accurate evaluation. See some examples in \cite{V2}, \cite{V1}. By the way, in~\cite{V1} we get results whose numerical error is at most 0.05\%. Optimistically speaking, the 1.7\% could even be handled by our techniques.

Neither the hyperbolic nor the Euclidean geometry have positive sectional curvatures. Lawson surfaces have chance to be isoperimetric in some lens spaces $L(p,q)$, even though they are just the quotient of $\S^3$ by a finite cyclic group, which however gives bidirectional translations. Indeed, $\Gamma$ in Sect.~\ref{prelim} could have been defined just as $\langle T_1,T_2\rangle$ instead of $\langle T_1,T_2,T_3\rangle$. In this case $\aaa$ in Fig.~\ref{aaa}(a) will rise indefinitely in the vertical direction. The same holds for the examples discussed in~\cite{Ri}, \cite{R1} but since these works inspired ours then we have considered a cubic tesselation.

Regarding $L(p,q)$ we mention the recent result~\cite{V}, which excludes Lawson surfaces for infinitely many cases. There the author solves the isoperimetric problem for lens spaces with a large fundamental group: either geodesic spheres or tori of revolution about geodesics. He also proves that the only candidates in $L(3,q)$ are geodesic spheres of flat tori, $q=1,\,2$.

Notice that a cubic tesselation exists in $\S^3$ as we did for $\H^3$ but without a group that acts totally discontinuously. Hence we cannot have translations because of fixed points. In order to see this, first consider the tiling in Fig.~\ref{3planes&tiling45}(b) and the corresponding tiling of $\S^2$ depicted in Fig.~\ref{sc}(a). Of course, compactness implies there are six spherical squares and the south pole is centred at the bottom square, the darkest one in Fig.~\ref{sc}(a). Any edge of the tiling belongs to a plane of $\R^3$ given by $x_i=\pm\,x_j$, $i\ne j$. Hence continuous displacements that keep the tiling must come from $90^\circ$-rotations about the fundamental axes, and any such rotation will have two fixed points.

For the corresponding cubic tesselation of $\S^3\subset\R^4$ consider $(0,0,0,\pm1)$ as the north $\cN$ and south $\cS$ poles, respectively. The stereographic projection $\S^3\setminus\{\cN\}\to\R^3$ that keeps the equator $\S^2\times\{0\}$ will take its inside to the unit ball in $\R^3$, so that $\cS$ corresponds to the origin. This ball is shown by grid lines in Fig.~\ref{sc}(b), and without loss of generality our tesselation contains the spherical cube centred at $\cS$ and depicted there. It has dihedral angle $2\pi/3$, hence by spherical reflections we get eight cubes that cover the whole $\S^3$. Now take any spherical displacement that keeps the tesselation. The origin will slide along a fundamental axis but the spherical cubes that are crossed by the other axes will remain invariant. Hence we shall get fixed points in the three-dimensional case as well.

\begin{figure}[ht!]
\center
\includegraphics[scale=0.360]{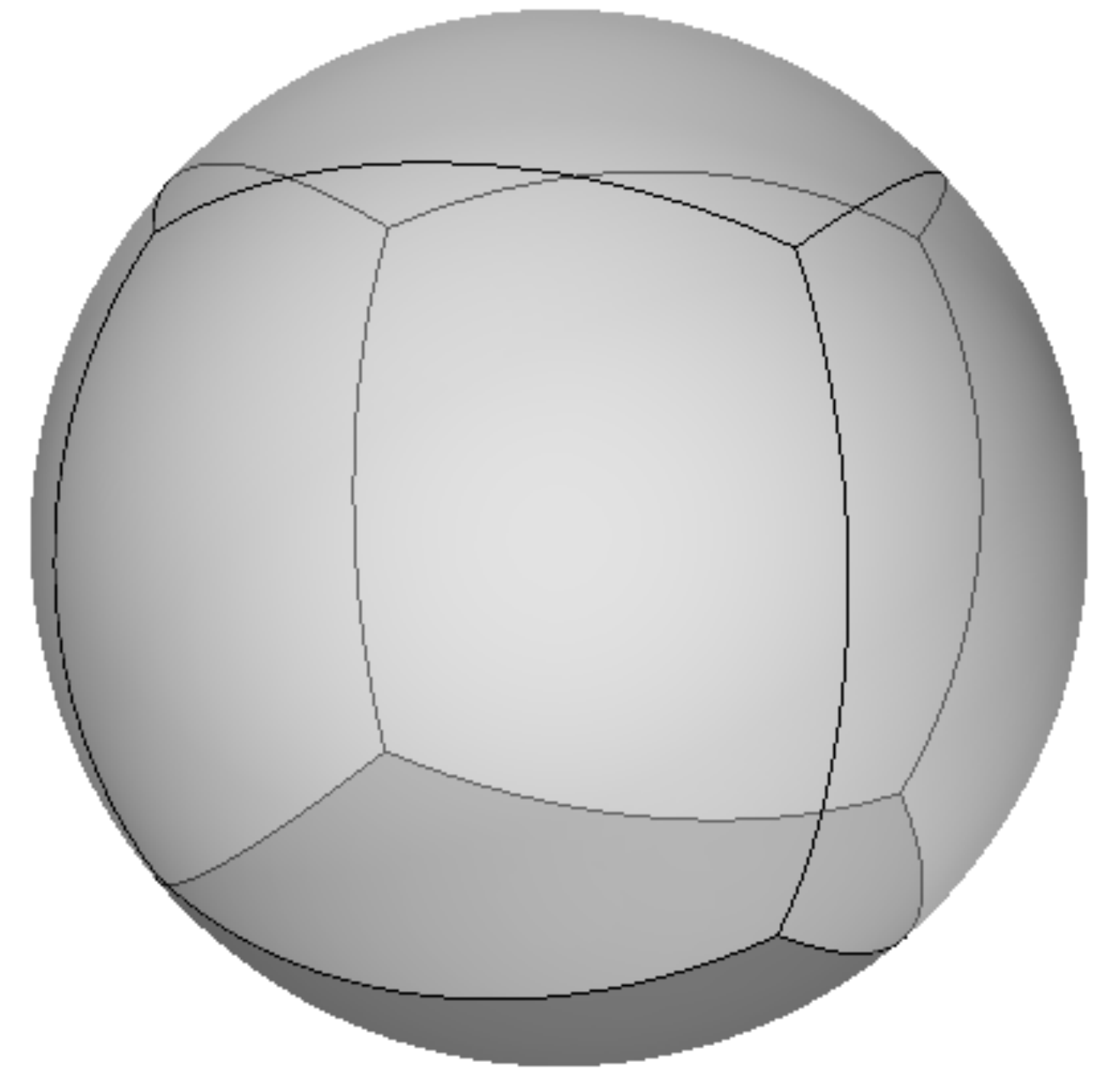}\hfill
\includegraphics[scale=0.205]{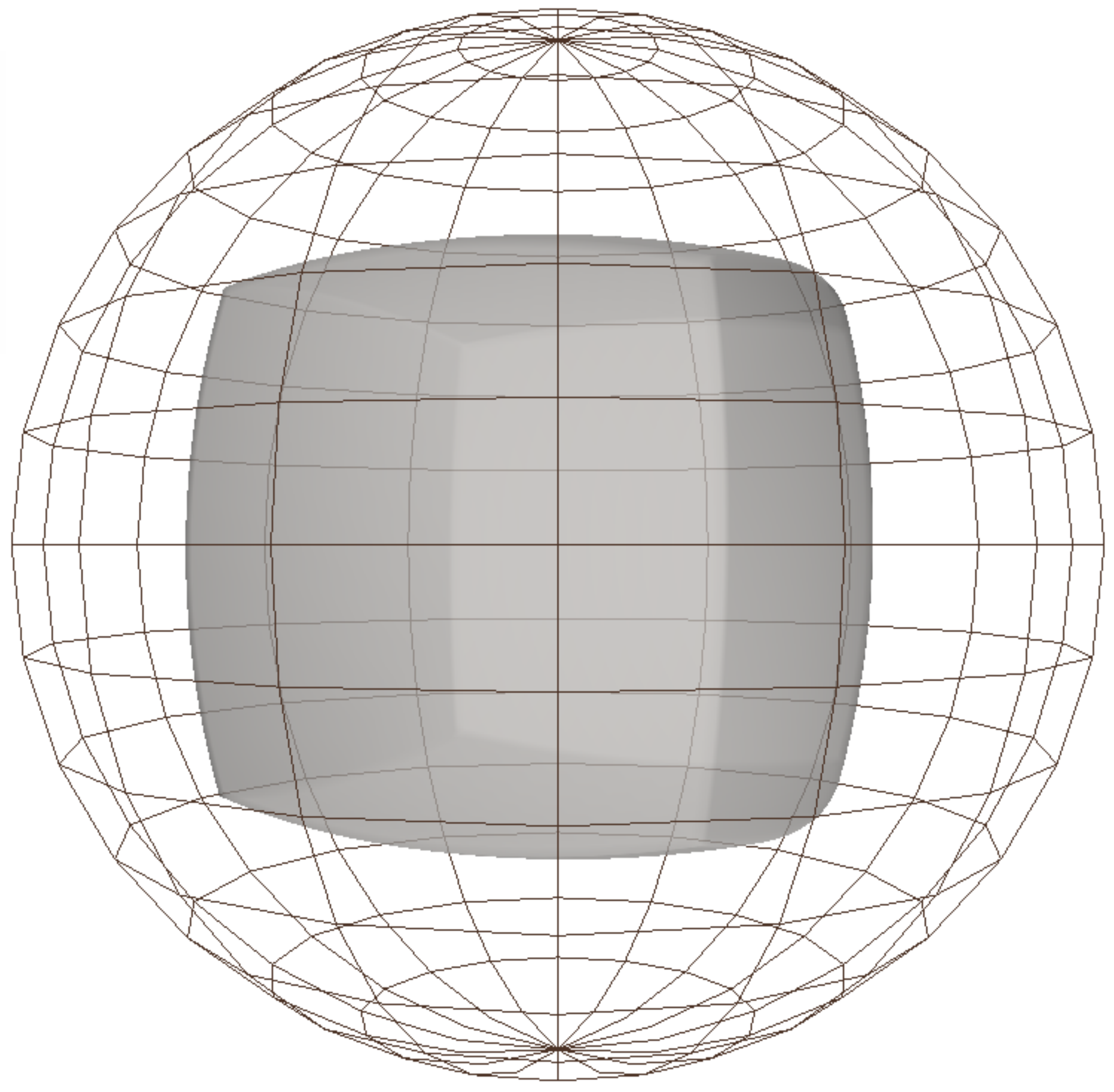}

\centerline{(a)\hfil\hfil(b)}

\caption{The unit ball containing the spherical cube with dihedral angle $2\pi/3$.}
\label{sc}
\end{figure}

As a matter of fact the unit ball can be viewed as a spherical cube of dihedral angle $\pi$, so that together with its reflection in $\S^2$ one gets a tesselation of $\S^3$ with only two cubes. However, for a new space form we must take the quotient by the group $\langle id,\mathfrak{A}\rangle$ (identity and antipodal maps). Hence the quotient will have a non-orientable boundary, which is the projective plane. However, the interest in the isoperimetric problem relies mostly on spaces whose geometry fits in the following sequence of generalisations: from Euclidean to spaces forms, from these to homogeneous spaces, and from these to Lie groups in general.

This has to do with the amount of results available for these geometries. For instance, in homogeneous spaces we still count on the maximum principle for CMC surfaces, and on the fact that they are minima of area under volume constraint. The {\it Surface Evolver} works in any finite dimension and it displays three-dimensional projections specified by the user. But of course our understanding and our control of the numerical computation will always work better in the third dimension. There the richest homogeneous spaces have four-dimensional isometry group, and for them much about CMC surfaces is already known. See~\cite{J} for a good start.

\section*{Acknowledgements}
We thank Prof K.Gro\ss e-Brauckmann at the Maths Dept of the University of Darmstadt for his valuable help in this work. This was during his stay in Brazil in 2010, supported by S\~ao Paulo Research Foundation (FAPESP) proc.09/15408-0. The first author was partially supported by FAPESP proc.16/23746-6.

\bibliographystyle{spbasic}      
\bibliography{mcs}

\end{document}